\pdfoutput=1
\RequirePackage{ifpdf}
\ifpdf 
\documentclass[pdftex]{sigma}
\else
\documentclass{sigma}
\fi

\numberwithin{equation}{section}

\newtheorem{Theorem}{Theorem}[section]
\newtheorem{Corollary}[Theorem]{Corollary}
\newtheorem{Lemma}[Theorem]{Lemma}
 { \theoremstyle{definition}
\newtheorem{Remark}[Theorem]{Remark} }

\usepackage{tikz}
\usepackage{mathtools}
\usepackage{bm}
\usepackage{mathrsfs}

\usetikzlibrary{arrows}
\usetikzlibrary{positioning}
\usetikzlibrary{decorations.pathmorphing}
\usetikzlibrary{decorations.markings}
\def\arrowhead{angle 90}
\tikzset{>=\arrowhead}
\pgfarrowsdeclaredouble{<<}{>>}{\arrowhead}{\arrowhead}
\pgfarrowsdeclaretriple{<<<}{>>>}{\arrowhead}{\arrowhead}
\pgfarrowsdeclaredouble{<<<<}{>>>>}{<<}{>>}
\tikzstyle{G}=[draw, circle, minimum size=1em, scale=1, inner sep=2pt]
\tikzstyle{B}=[draw,circle,fill=black,scale=1]
\tikzstyle{H}=[draw,circle,fill=white,scale=1]
\tikzstyle{F}=[draw, rectangle, minimum width=1em, minimum height=1em, scale=1]
\tikzstyle{every picture}=[scale=1,baseline=(current bounding box.south)]

\newcommand{\ds}{\displaystyle}
\newcommand{\q}{{\mathsf q}}
\newcommand{\p}{{\mathsf p}}

\def\EXP{\textrm{{\large e}}}
\newcommand{\ii}{i}

\newcommand{\x}{{\boldsymbol{x}}}
\newcommand{\y}{{\boldsymbol{y}}}

\newcommand{\bu}{{\boldsymbol{u}}}
\newcommand{\bv}{{\boldsymbol{v}}}

\newcommand{\olW}{\overline{W}}

\def\re{\operatorname{Re}}
\def\im{\operatorname{Im}}

\newcommand{\veca}{\bm{a}}
\newcommand{\vecb}{\bm{b}}
\newcommand{\vecs}{\bm{s}}
\newcommand{\vect}{\bm{t}}
\newcommand{\vecc}{\bm{c}}
\newcommand{\vecd}{\bm{d}}
\newcommand{\vecx}{\bm{x}}
\newcommand{\vecw}{\bm{w}}
\newcommand{\tveca}{\tilde{\bm{a}}}
\newcommand{\tvecb}{\tilde{\bm{b}}}
\newcommand{\tvecs}{\tilde{\bm{s}}}
\newcommand{\tvect}{\tilde{\bm{t}}}
\newcommand{\bveca}{\bar{\bm{a}}}
\newcommand{\bvecb}{\bar{\bm{b}}}
\newcommand{\bvecs}{\bar{\bm{s}}}
\newcommand{\bvect}{\bar{\bm{t}}}
\newcommand{\vecz}{\bm{z}}
\newcommand{\vecy}{\bm{y}}

\DeclarePairedDelimiter\floor{\lfloor}{\rfloor}

\begin{document}

\allowdisplaybreaks

\newcommand{\arXivNumber}{1704.03159}

\renewcommand{\thefootnote}{}

\renewcommand{\PaperNumber}{013}

\FirstPageHeading

\ShortArticleName{Elliptic Hypergeometric Sum/Integral Transformations and Supersymmetric Lens Index}

\ArticleName{Elliptic Hypergeometric Sum/Integral\\ Transformations and Supersymmetric Lens Index\footnote{This paper is a~contribution to the Special Issue on Elliptic Hypergeometric Functions and Their Applications. The full collection is available at \href{https://www.emis.de/journals/SIGMA/EHF2017.html}{https://www.emis.de/journals/SIGMA/EHF2017.html}}}

\Author{Andrew P.~KELS~$^\dag$ and Masahito YAMAZAKI~$^\ddag$}

\AuthorNameForHeading{A.P.~Kels and M.~Yamazaki}

\Address{$^\dag$~Institute of Physics, University of Tokyo, Komaba, Tokyo 153-8902, Japan}
\EmailD{\href{mailto:andrew.p.kels@gmail.com}{andrew.p.kels@gmail.com}}

\Address{$^\ddag$~Kavli Institute for the Physics and Mathematics of the Universe (WPI), University of Tokyo,\\
\hphantom{$^\ddag$}~Chiba 277-8583, Japan}
\EmailD{\href{mailto:masahito.yamazaki@ipmu.jp}{masahito.yamazaki@ipmu.jp}}

\ArticleDates{Received April 24, 2017, in f\/inal form February 02, 2018; Published online February 16, 2018}

\Abstract{We prove a pair of transformation formulas for multivariate elliptic hypergeometric sum\slash integrals associated to the $A_n$ and $BC_n$ root systems, generalising the formulas previously obtained by Rains. The sum/integrals are expressed in terms of the lens elliptic gamma function, a generalisation of the elliptic gamma function that depends on an additional integer variable, as well as a complex variable and two elliptic nomes. As an application of our results, we prove an equality between $S^1\times S^3/\mathbb{Z}_r$ supersymmetric indices, for a pair of four-dimensional $\mathcal{N}=1$ supersymmetric gauge theories related by Seiberg duality, with gauge groups ${\rm SU}(n+1)$ and ${\rm Sp}(2n)$. This provides one of the most elaborate checks of the Seiberg duality known to date. As another application of the $A_n$ integral, we prove a star-star relation for a two-dimensional integrable lattice model of statistical mechanics, previously given by the second author.}


\Classification{33C67; 33E20; 81T60; 81T13; 82B23; 16T25}

\renewcommand{\thefootnote}{\arabic{footnote}}
\setcounter{footnote}{0}

\section{Introduction}

Elliptic hypergeometric series and integrals comprise the top-level class of hypergeometric functions, and provide generalisations of many of the well-known classical and basic hypergeometric functions and their corresponding identities. In contrast to the latter classical and basic counterparts whose study was initiated centuries ago, the area of elliptic hypergeometric functions has only been developed in relatively recent times, following their initial discovery from the Boltzmann weights of the integrable fused RSOS models of statistical mechanics \cite{DJKMO,Frenkel1997}. Central to this paper are the elliptic hypergeometric integrals that are expressed in terms of the elliptic gamma function, and satisfy many remarkable identities \cite{Spiridonov-essays}, some of which have been found to have important applications in dif\/ferent areas of mathematical physics.

One of these areas is exactly solved models, where the elliptic beta integral \cite{SpiridonovEBF}, a central identity for the theory of elliptic hypergeometric functions, is known to be equivalent to a Yang--Baxter equation \cite{Bazhanov:2010kz} (more specif\/ically a star-triangle relation), which is a fundamental identity for integrability of two-dimensional lattice models of statistical mechanics \cite{Baxter:1982zz}. Specif\/ically, the Yang--Baxter equation implies that the row-to-row transfer matrices of the lattice model commute, and following the method of Baxter \cite{Baxter:1972hz}, this can be used to solve for the partition function in the thermodynamic limit. Such lattice models related to elliptic hypergeometric integrals are quite general \cite{Bazhanov:2016ajm,Bazhanov:2007mh,Bazhanov:2007vg,Bazhanov:2010kz,GahramanovKels,Gahramanov:2015cva,Kels:2013ola,Kels:2015bda,Spiridonov:2010em}, and reduce to many important integrable lattice models of ``Ising type'' (e.g., \cite{AuYang:1987zc,Baxter:1987eq,Faddeev:1993pe,Fateev:1982wi,Kashiwara:1986tu,Volkov:1992uv}) as special limiting cases.

The above lattice models satisfying the star-triangle relation, involve interactions between single-component spins. There exist also multi-component spin models that satisfy a so-called star-star relation, another fundamental identity for integrability of lattice models \cite{Baxter:1997tn}. In this case a solution of the star-star relation introduced by Bazhanov and Sergeev \cite{Bazhanov:2011mz}, was recently shown~\cite{Bazhanov:2013bh} to be equivalent to a special case of Rains' multivariate transformation formula associated to the $A_n$ root system, due to Rains \cite{RainsT}. This example establishes another link between integrable lattice models and elliptic hypergeometric integrals.

Another important connection arises between identities of elliptic hypergeometric integrals, and Seiberg duality \cite{Seiberg:1994pq} of the indices of four-dimensional supersymmetric gauge theories. This was f\/irst observed by Dolan and Osborn \cite{Dolan:2008qi}, who showed that supersymmetric indices for a pair of Seiberg-dual theories on $S^1\times S^3$ \cite{Kinney:2005ej}, are equivalent as a consequence of Rains transformation formulas \cite{RainsT}. This connection provides a rigorous verif\/ication of the proposed dualities for supersymmetric gauge theories, by matching the equivalence of supersymmetric indices with the mathematically proven identities of elliptic hypergeometric integrals. This particular connection is quite powerful, as in principle it provides a way to generate a large number of complicated, and generally new, identities of elliptic hypergeometric integrals (and in some cases the Yang--Baxter equation), from systematic analysis of dualities among supersymmetric gauge theories~\cite{Spiridonov:2009za,Spiridonov:2011hf}.

One interesting direction is to generalise the above indices on $S^1\times S^3$, by replacing the~$S^3$ with a lens space $S^3/\mathbb{Z}_r$~\cite{Benini:2011nc}. The resulting expressions for the indices, depend on sets of both complex and integer variables, and involve a summation over some discrete variable, as well as integration, while for $r=1$, the expressions reduce to known elliptic hypergeometric integrals~\cite{Dolan:2008qi}. Consequently, such an expression that involves both the summation and integration will be referred to here as an ``elliptic hypergeometric sum/integral'' (there has also been proposed the name ``raref\/ied elliptic hypergeometric integral''~\cite{rarified}). The sum/integrals are expressed in terms of a generalisation of the elliptic gamma function~\cite{Ruijsenaars:1997:FOA}, known as the lens elliptic gamma function~\cite{Benini:2011nc,Kels:2015bda}. Compared to the elliptic gamma function, the lens elliptic gamma function depends on an extra integer parameter $r=1,2,\ldots$, and an extra integer variable~$m$ mod~$r$, and reduces to the regular elliptic gamma function for the case $r=1$.

Equivalence of Seiberg dual indices on $S^1\times S^3/\mathbb{Z}_r$, implies corresponding identities between dif\/ferent elliptic hypergeometric sum/integrals. The simplest example of such an identity is the elliptic beta sum/integral proven by the f\/irst author \cite{Kels:2015bda}, which corresponds to ``electric-magnetic'' duality between two particular $\mathcal{N}=1$ theories \cite{GahramanovKels}. This elliptic beta sum/integral depends on six integer variables, as well as six complex variables and two complex elliptic nomes, and reduces to Spiridonov's elliptic beta integral \cite{SpiridonovEBF} for $r=1$. Spiridonov has also recently proven sum/integral evaluation formulas associated to the $BC_n$ root system~\cite{rarified}, and considered further mathematical properties of these sum/integral expressions, for example, deriving the sum/integral analogue of the elliptic Gauss hypergeometric equation~\cite{Spiridonov-essays}.

As expected, the above sum/integral identities also have relevance to integrable lattice mo\-dels. In this context, the second author found a solution of the star-star relation \cite{Yamazaki:2013nra}, that generalises the multi-spin model of Bazhanov and Sergeev~\cite{Bazhanov:2011mz}, to the case of continuous, as well as discrete, spin variables. The simplest case of this model corresponds to the above elliptic beta sum/integral~\cite{Kels:2015bda}. The star-star relation~\cite{Yamazaki:2013nra} provides one of the most general solutions of the Yang--Baxter equation known in the literature, and one of the motivations of this paper was to obtain the corresponding elliptic hypergeometric sum/integral transformation formula for the general case.

The main result of this paper, is proving a pair of new sum/integral transformation formulas that are associated respectively to the $A_n$ and $BC_n$ root systems. It is shown how these transformation formulas imply lens index duality for supersymmetric gauge theories on $S^1\times S^3/\mathbb{Z}_r$~\cite{Benini:2011nc}, and imply the above-mentioned star-star relation~\cite{Yamazaki:2013nra} (in the $A_n$ case). The sum/integral transformation formulas generalise the corresponding transformation formulas of Rains \cite{RainsT}, where the latter are equivalent to the choice $r=1$. As expected, the $BC_n$ sum/integral transformation formula derived in this paper, reduces to Spiridonov's $BC_n$ sum/integral identity~\cite{rarified} as a~special case.

The method of proof for the sum/integral formulas basically follows from Rains' proofs~\cite{RainsT} of the $r=1$ cases. That is, the transformations are f\/irst proven for some special choice of the variables, for which the transformations can be written in the form of a determinant of univariate sum/integrals of (lens) theta functions. Taking limits of these special cases provides a dense set of cases for which the transformations hold, and thus the transformations hold in general. There are some essential modif\/ications in the details of the proofs, due to the appearance of the additional integer variables, and also due to the use of the lens elliptic gamma function, which has a non-trivial normalisation factor given in terms of multiple Bernoulli polynomials~\cite{Narukawa2004247}. Otherwise the steps of the proofs are quite analogous to the case $r=1$.

The paper is arranged as follows. Section \ref{sec:defs} def\/ines the lens elliptic gamma function, and gives a number of identities that it satisf\/ies, which are used throughout the paper. Sections~\ref{sec:antrans} and~\ref{sec:bcntrans} present the respective $A_n$ and $BC_n$ sum/integral transformation formulas, along with their proofs. These sections also contain an overview of the related elliptic hypergeometric integrals that are obtained as special cases of the sum/integral transformations. Section~\ref{sec:SUSY} discusses a supersymmetric gauge theory interpretation of the~$A_n$ and~$BC_n$ transformation formulas as Seiberg dualities. Finally, Section~\ref{sec:YBE} introduces the aforementioned lattice model with multi-component real and integer valued spin variables~\cite{Yamazaki:2013nra}, and shows that the star-star relation of the model is equivalent to a special case of the $A_n$ sum/integral transformation formula.

\section{Def\/initions}\label{sec:defs}

The multiple Bernoulli polynomials $B_{n,k}(z;\omega_1,\ldots,\omega_n)$ are def\/ined through the generating function
\begin{gather*}
\frac{x^n \EXP^{zx}}{\prod\limits_{j=1}^n(\EXP^{\omega_j x}-1)}=\sum_{k=0}^\infty B_{n,k}(z;\omega_1,\ldots,\omega_n) \frac{x^k}{k!},
\end{gather*}
where $z\in\mathbb{C}$, and $\omega_1,\ldots,\omega_n\in\mathbb{C}-\{0\}$.

These functions previously appeared in relation to the modular properties of multiple gamma functions \cite{Narukawa2004247}. For this paper only a particular multiple Bernoulli polynomial $B_{3,3}(z;\omega_1,\omega_2,\omega_3)$, is needed, which is given explicitly by
\begin{gather}
B_{3,3}(z;\omega_1,\omega_2,\omega_3)= \frac{z^3}{\omega_1\omega_2\omega_3}- \frac{3z^2\sum\limits_{i=1}^3\omega_i}{2\omega_1\omega_2\omega_3}+\frac{z\left(\sum\limits_{i=1}^3\omega_i^2+3\sum_{1\leq i<j\leq 3}\omega_i\omega_j\right)}{2\omega_1\omega_2\omega_3}\nonumber\\
\hphantom{B_{3,3}(z;\omega_1,\omega_2,\omega_3)=}{} -\frac{\left(\sum\limits_{i=1}^3\omega_i\right)\left(\sum\limits_{1\leq i<j\leq3}\omega_i\omega_j\right)}{4\omega_1\omega_2\omega_3}.\label{bernoulli}
\end{gather}

Let us also introduce the two complex parameters $\sigma, \tau \in \mathbb{C}$ that satisfy
\begin{gather*}
\im(\sigma),\; \im(\tau)>0,
\end{gather*}
and def\/ine $R(z;\sigma,\tau)$, and $R_2(z,m;\sigma,\tau)$, as the following combinations of \eqref{bernoulli}:
\begin{gather}
R(z;\sigma,\tau):=\frac{B_{3,3}(z;\sigma,\tau,-1)+B_{3,3}(z-1;\sigma,\tau,-1)}{12},\nonumber\\
 R_2(z,m;\sigma,\tau):= R(z+m\sigma;r\sigma,\sigma+\tau)+R(z+(r-m)\tau;r\tau,\sigma+\tau)\nonumber\\
\hphantom{R_2(z,m;\sigma,\tau)}{} =\ds\frac{(\sigma+\tau-2z)(2z^2-2z(\sigma+\tau)+\sigma\tau(r^2+6(m-r)m)+1)}{24r\sigma\tau}\nonumber\\
\hphantom{R_2(z,m;\sigma,\tau)=}{} -\frac{(\sigma-\tau)(2m-r)(m-r)m}{12r},\label{r2def}
\end{gather}
where $z\in\mathbb{C}$, $m\in\mathbb{Z}$, and
\begin{gather*}
r=1,2,\ldots,
\end{gather*}
is an integer parameter.

The lens elliptic gamma function is def\/ined here as \cite{Benini:2011nc,GahramanovKels,Kels:2015bda,Razamat:2013opa}
\begin{gather}\label{legf2}
\Gamma(z,m;\sigma,\tau):=\EXP^{\phi_e(z, m;\sigma,\tau)}\gamma(z,m;\sigma,\tau),\qquad z\in\mathbb{C}, \qquad m\in \{0,1,\ldots,r-1\},
\end{gather}
where
\begin{gather}
 \phi_e(z,m;\sigma,\tau) = 2\pi\ii\left(R_2\left(z,0; \sigma-\frac{1}{2},\tau+\frac{1}{2}\right)-R_2\left(z,m;\sigma-\frac{1}{2},\tau+\frac{1}{2}\right)\right)\nonumber\\
\hphantom{\phi_e(z,m;\sigma,\tau)}{} =\ds 2\pi\ii\left(R_2(z,0;\sigma,\tau)+R_2\left(0,m; \frac{1}{2},-\frac{1}{2}\right)-R_2(z,m;\sigma,\tau)\right),\label{ellnorm}
\end{gather}
and
\begin{gather}\label{littlegamma}
\gamma(z,m;\sigma,\tau)=\prod_{j,k=0}^\infty\frac{1-\EXP^{-2\pi\ii z}\p^{-m}(\p\q)^{j+1}\p^{r(k+1)}}{1-\EXP^{2\pi\ii z}\p^{m}(\p\q)^j\p^{rk}}\frac{1-\EXP^{-2\pi\ii z}\q^{-r+m}(\p\q)^{j+1}\q^{r(k+1)}}{1-\EXP^{2\pi\ii z}\q^{r-m}(\p\q)^j\q^{rk}}.
\end{gather}
The elliptic nomes in \eqref{littlegamma} are def\/ined in terms of $\sigma$ and $\tau$ as
\begin{gather*}
\p=\EXP^{2\ii\pi\sigma},\qquad\q=\EXP^{2\ii\pi\tau}.
\end{gather*}

The lens elliptic gamma function \eqref{legf2} is periodic in the complex and integer arguments respectively, satisfying
\begin{gather}
\Gamma(z+2kr,m;\sigma,\tau)=\Gamma(z,m;\sigma,\tau),\qquad\Gamma(z,m+kr;\sigma,\tau)=\Gamma(z,m;\sigma,\tau),\qquad k\in\mathbb{Z}.
\label{periodic}
\end{gather}
Accordingly, throughout the paper the notation $a\mbox{ mod }r$, for an integer $a$, is always taken to be the corresponding element of $\{0,1,\ldots,r-1\}$, as in \eqref{legf2}. Note also that $\gamma(z+k,m;\sigma,\tau)=\gamma(z,m;\sigma,\tau)$ for integer $k$, and
the $2r$-periodicity of $\Gamma(z,m; \sigma, \tau)$ comes from the factor $\phi_e$.

For $r=1$ (in which case we may take $m=0$), \eqref{legf2} is just the usual elliptic gamma function~\cite{Ruijsenaars:1997:FOA}, def\/ined as
\begin{gather}\label{egf2}
\ds\left.\Gamma(z,m;\sigma,\tau)\,\right|_{(r=1)}=\Gamma_{1}(z;\sigma,\tau)=\ds\prod_{j,k=0}^\infty\frac{1-\EXP^{-2\pi\ii z}\p^{j+1}\q^{k+1}}{1-\EXP^{2\pi\ii z}\p^j \q^k}.
\end{gather}
The lens elliptic gamma function \eqref{legf2} may be written as a product of two regular elliptic gamma functions \eqref{egf2}, as
\begin{gather*}
\Gamma(z,m;\sigma,\tau)=\EXP^{\phi_e(z,m;\sigma,\tau)} \Gamma_{1}(z+\sigma m;r\sigma,\sigma+\tau) \Gamma_{1}(z+\tau(r-m);r\tau,\sigma+\tau).
\end{gather*}
In this paper, the lens elliptic gamma function \eqref{legf2} will usually be written as $\Gamma(z,m)$, where{\samepage
\begin{gather*}
\Gamma(z,m):=\Gamma(z,m;\sigma,\tau),
\end{gather*}
with implicit dependence on the parameters $\sigma$ and $\tau$.}

For integers $m\in\{0,1,\ldots,r-1\}$, the poles and zeroes of lens elliptic gamma function \eqref{legf2} are respectively located at the points
\begin{gather*}
z=-(\sigma+\tau)j-\sigma(rk+m)+n,\; -(\sigma+\tau)j-\tau(r(k+1)-m)+n, \\
 z=(\sigma+\tau)(j+1)+\sigma(r(k+1)-m)+n,\; (\sigma+\tau)(j+1)+\tau(rk+m)+n,
\end{gather*}
where $j,k=0,1,\ldots$, and $n\in\mathbb{Z}$.

The lens elliptic gamma function satisf\/ies some useful relations:
\begin{gather*}
\Gamma((\sigma+\tau)-z,-m)\Gamma(z,m)=1,\\
\Gamma(z+n\sigma,m-n)=\Gamma(z,m)\prod_{j=0}^{n-1}\theta_1(z+j\sigma,m-j),\\
\Gamma(z+n\tau,m+n)=\Gamma(z,m)\prod_{j=0}^{n-1}\theta_2(z+j\tau,m+j),
\end{gather*}
for integers $n=1,2,\dots$. Here the theta functions are def\/ined as
\begin{gather}
\theta_1(z,m)=\EXP^{\phi_1(z,m)} \theta\big(\EXP^{-2\pi\ii z}\EXP^{2\pi\ii\tau m}\,|\,\EXP^{2\pi\ii\tau r}\big),\nonumber\\
\theta_2(z,m)=\EXP^{\phi_2(z,m)} \theta\big(\EXP^{2\pi\ii z}\EXP^{2\pi\ii\sigma m}\,|\,\EXP^{2\pi\ii\sigma r}\big),\label{lthtdef}
\end{gather}
where the normalisation factors are
\begin{gather}
\phi_1(z,m)=\phi_e(z+\sigma,m-1;\sigma,\tau)-\phi_e(z,m;\sigma,\tau) \nonumber\\
\hphantom{\phi_1(z,m)}{} =\frac{\pi\ii}{12r}\big(3(r+1-2m)(2z+1)-(r^2-1)(\sigma-\tau-1)-6m(r-m)(\tau+1)\big), \nonumber\\
\phi_2(z,m) =\phi_e(z+\tau,m+1;\sigma,\tau)-\phi_e(z,m;\sigma,\tau)\nonumber\\
\hphantom{\phi_2(z,m)}{} =\frac{-\pi\ii}{12r}\big(3(r-1-2m)(2z-1)-(r^2-1)(\sigma-\tau-1)+6m(r-m)(\sigma-1)\big),\!\!\!\label{lthtnorm}
\end{gather}
with $\phi_e$ as def\/ined in \eqref{ellnorm}, and where $\theta(z\,|\,\q)$ is the usual theta function
\begin{gather*}
\theta(z\,|\,\q)=\left(z;\q\right)_\infty \left(\frac{\q}{z};\q\right)_\infty,
\qquad
(z;\q)_\infty=\prod_{j=0}^\infty\big(1-z\q^j\big).
\end{gather*}
The theta functions $\theta_1$, and $\theta_2$, in \eqref{lthtdef}, have non-trivial dependence on both of the parameters~$\sigma$, and~$\tau$, through the normalisation functions~\eqref{lthtnorm}.

The theta functions \eqref{lthtdef}, each satisfy the same periodicities \eqref{periodic} as the lens elliptic gamma function, i.e., for any integer~$k$
\begin{gather*}
\theta_1(z+2kr,m)=\theta_1(z,m),\qquad\theta_1(z,m+kr)=\theta_1(z,m), \\
\theta_2(z+2kr,m)=\theta_2(z,m),\qquad\theta_2(z,m+kr)=\theta_2(z,m).
\end{gather*}
The theta functions \eqref{lthtdef} also satisfy
\begin{gather*}
\theta_1(-z,-m)=-\theta_1(z,m)\EXP^{-2\pi\ii(z-m)/r},\qquad
\theta_2(-z,-m)=-\theta_2(z,m)\EXP^{-2\pi\ii(z-m)/r},
\end{gather*}
and
\begin{gather*}
\theta_1(z+n\tau,m+n)=\theta_1(z,m) \EXP^{-n\pi\ii(2z+(n-1)\tau+r-2m-n+1)/r},\\
\theta_2(z+n\sigma,m-n)=\theta_2(z,m) \EXP^{-n\pi\ii(2z+(n-1)\sigma+r-2m+n-1)/r},\\
\theta_1(z+rn\tau,m)=\theta_1(z,m) \EXP^{-n\pi\ii(2z+\tau(rn-1)+1)},\\
\theta_2(z+rn\sigma,m)=\theta_2(z,m) \EXP^{-n\pi\ii(2z+\sigma(rn-1)+1)},
\end{gather*}
for integers $n$.

In this paper, a set of complex variables $t_1,\ldots,t_n\in\mathbb{C}$ will frequently be represented as a~vector
\begin{gather*}
\vect=(t_1,\ldots,t_n),
\end{gather*}
and addition with a complex number $\gamma\in\mathbb{C}$ is given by
\begin{gather*}
\gamma+\vect:=(\gamma+t_1,\ldots,\gamma+t_n).
\end{gather*}
An analogous notation also applies to sets of integer variables $a_1,\ldots,a_n\in\mathbb{Z}$, and addition with integers.

\section[The $A_n\leftrightarrow A_m$ transformation]{The $\boldsymbol{A_n\leftrightarrow A_m}$ transformation}\label{sec:antrans}

\subsection{Main theorem}

Let us introduce the complex variables $\sigma$, $\tau$, $t_i$, $s_i$, and integer variables $a_i$, $b_i$, for $i=0,1,\ldots,m+n+1$, satisfying
\begin{gather}
\im(\sigma),\im(\tau)>0, \nonumber\\
\sum_{i=0}^{m+n+1} (t_i+s_i)\equiv(m+1)(\sigma+\tau)\ (\textrm{mod }2r), \qquad \sum_{i=0}^{m+n+1} (a_i+b_i)\equiv 0 \ (\textrm{mod }r).\label{balancing}
\end{gather}
In terms of these variables, we def\/ine $I^m_{A_n}(Z,Y\,|\,\vect,\veca;\vecs,\vecb)$ as the following elliptic hypergeometric sum/integral
\begin{gather}
I^m_{A_n}(Z,Y\,|\,\vect,\veca;\vecs,\vecb)\ds=\frac{\lambda^n}{(n+1)!}\sum_{\substack{y_0,\ldots,y_{n-1}=0 \\ \sum\limits_{i=0}^{n}y_i=Y}}^{r-1}\int_{\sum\limits_{i=0}^nz_i=Z}\Delta^m_{A_n}(\vecz,\vecy;\vect,\veca;\vecs,\vecb)\prod_{i=0}^{n-1} {\rm d}z_i,\label{AnIntDef}
\end{gather}
where $m,n=0,1,\ldots$,
\begin{gather*}
\lambda=\big(\p^r;\p^r\big)_\infty\big(\q^r;\q^r\big)_\infty,\\
\Delta^m_{A_n}(\vecz,\vecy;\vect,\veca;\vecs,\vecb)=\frac{\prod\limits_{i=0}^n\prod\limits_{j=0}^{m+n+1}\Gamma(t_j+z_i,a_j+y_i) \Gamma(s_j-z_i,b_j-y_i)}{\prod\limits_{0\leq i<j\leq n}\Gamma(z_i-z_j,y_i-y_j) \Gamma(z_j-z_i,y_j-y_i)},
\end{gather*}
and
\begin{alignat}{4}
& \vecz=(z_0,z_1,\ldots,z_{n}),\qquad&& \vect=(t_0,t_1,\ldots,t_{m+n+1}),\qquad&& \vecs=(s_0,s_1,\ldots,s_{m+n+1}),& \nonumber\\
& \vecy=(y_0,y_1,\ldots,y_{n}),\qquad&&\veca=(a_0,a_1,\ldots,a_{m+n+1}),\qquad&& \vecb=(b_0,b_1,\ldots,b_{m+n+1}).&\label{anvariables}
\end{alignat}
Due to the periodicities of the lens elliptic gamma function \eqref{periodic}, the condition on the summation variables in \eqref{AnIntDef} is to be understood as $\sum\limits_{i=0}^ny_i=Y \; (\textrm{mod }r)$, and the condition on the integration variables is to be understood as $\sum\limits_{i=0}^nz_i=Z \; (\textrm{mod }2r)$. However in the following we will avoid writing the latter $\textrm{mod }$ conditions on the summation and integration variables for conciseness.

For the values satisfying $\im(s_i)>\frac{\im(Z)}{n+1}>-\im(t_i)$, the contour in \eqref{AnIntDef} may be chosen to be $C^{n}$, where $C$ is a straight line that connects the two points $\ii\frac{\im(Z)}{n+1}$, and $1+\ii\frac{\im(Z)}{n+1}$. Otherwise the sum/integral \eqref{AnIntDef} is def\/ined by meromorphic continuation from the latter case, with appropriately chosen contours connecting the points $z_j=k_j\ii$, respectively to the points $z_j=1+k_j\ii$, where $k_j$ are real numbers, for $j=0,1,\ldots,n-1$.

The particular case $n=0$ of \eqref{AnIntDef} is given by
\begin{gather*}
I^m_{A_{0}}(Z,Y\,|\,\vect,\veca;\vecs,\vecb)=\prod_{j=0}^{m+1}\Gamma(t_j+Z,a_j+Y)\Gamma(s_j-Z,b_j-Y).
\end{gather*}
The integrand $\Delta^m_{A_n}(\vecz,\vecy;\vect,\veca,\vecs,\vecb)$ is obviously $r$-periodic in each the discrete variables $y_i$,~$a_i$,~$b_i$ (because of $r$-periodicity \eqref{periodic} of the lens elliptic gamma function~\eqref{legf2}). The integrand is also non-trivially periodic under the shift of any integration variable $z_i$ by $z_i+k_i$ for integers $k_i$, where $\sum\limits_{i=0}^n k_i=0$ (due to the condition $\sum\limits_{i=0}^nz_i=Z$). The latter periodicity in the integration variables $z_i$ follows from the balancing condition $\sum\limits_{i=0}^{m+n+1}(a_i+b_i)\equiv 0\, (\textrm{mod }r)$. Finally, the integrand satisf\/ies the usual $2r$-periodicity in the complex variables $t_i$ and $s_i$.

In the following let us def\/ine
\begin{gather}
T=\sum_{i=0}^{m+n+1}t_i,\qquad S=\sum_{i=0}^{m+n+1}s_i,\qquad A=\sum_{i=0}^{m+n+1}a_i,\qquad B=\sum_{i=0}^{m+n+1}b_i.
\end{gather}

The main result of this section is the following elliptic hypergeometric sum/integral transformation formula.

\begin{Theorem}\label{mainthm} The sum/integral \eqref{AnIntDef}, under the balancing condition \eqref{balancing}, satisfies
\begin{gather}\label{transdef}
I^m_{A_n}(Z,Y\,|\,\vect,\veca;\vecs,\vecb)=I^n_{A_m}(Z+T,Y+A\,|\,\tvect,\tveca;\tvecs,\tvecb)\prod_{i,j=0}^{m+n+1}\Gamma(t_i+s_j,a_i+b_j),
\end{gather}
where
\begin{gather}\label{transrule}
\tvect=-\vect,\qquad\tvecs=\sigma+\tau-\vecs,\qquad\tveca=-\veca,\qquad\tvecb=-\vecb.
\end{gather}
\end{Theorem}

\subsection{Corollaries}

Theorem \ref{mainthm} contains as special cases several existing results in the literature, which are summarised in the diagram below:
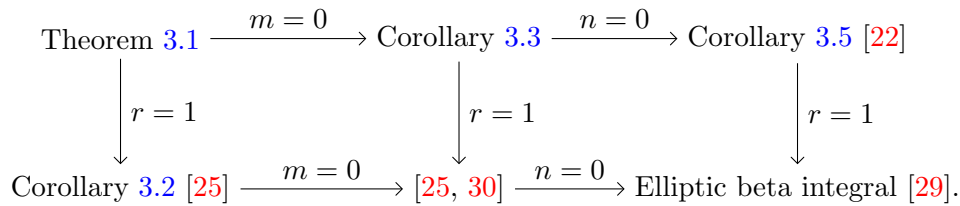
\begin{center}
\begin{tikzpicture}
\node (A) at (0, 2) {Theorem \ref{mainthm}};
\node (B) at (4.5, 2) {Corollary \ref{cor.mn0} };
\node (C) at (9, 2) {Corollary \ref{cor.Kels} \cite{Kels:2015bda}};
\node (D) at (0, 0) {Corollary \ref{cor.RainsT} \cite{RainsT}};
\node (E) at (4.5, 0) {\cite{RainsT, SpiridonovTheta}};
\node (F) at (9, 0) {Elliptic beta integral \cite{SpiridonovEBF}.};
\draw[->] (A) -- (B) node [midway, above] {$m=0$};
\draw[->] (B) -- (C) node [midway, above] {$n=0$};
\draw[->] (A) -- (D) node [midway, right] {$r=1$};
\draw[->] (B) -- (E) node [midway, right] {$r=1$};
\draw[->] (C) -- (F) node [midway, right] {$r=1$};
\draw[->] (D) -- (E) node [midway, above] {$m=0$};
\draw[->] (E) -- (F) node [midway, above] {$n=0$};
\end{tikzpicture}
\end{center}

First, the case $r=1$ of Theorem \ref{mainthm} is equivalent to the following $A_n\leftrightarrow A_m$ elliptic hypergeometric integral transformation formula proven by Rains \cite{RainsT}:

\begin{Corollary}[\cite{RainsT}]
\label{cor.RainsT}\begin{gather*}
\frac{\lambda^n}{(n+1)!}\int_{\sum\limits_{i=0}^nz_i=Z} \Delta^m_{A_n}(\vecz,\vecy;\vect,\vecs)\prod_{i=0}^{n-1}{\rm d}z_i = \prod_{i,j=0}^{m+n+1} \Gamma_1(t_i+t_j) \\
\qquad{} \times \frac{\lambda^m}{(m+1)!}\int_{\sum\limits_{i=0}^mz_i=Z}\Delta^n_{A_m}(\vecz,\vecy;\tvect,\tvecs)\prod_{i=0}^{m-1}{\rm d}z_i,
\end{gather*}
where
\begin{gather*}
\Delta^m_{A_n}(\vecz,\vecy;\vect,\vecs)=\frac{\prod\limits_{i=0}^n\prod\limits_{j=0}^{m+n+1}\Gamma_1(t_j+z_i) \Gamma_1(s_j-z_i)}{\prod\limits_{0\leq i<j\leq n}\Gamma_1(z_i-z_j) \Gamma_1(z_j-z_i)},
\end{gather*}
and $\Gamma_1(z):=\Gamma_1(z;\sigma,\tau)$ is the usual elliptic gamma function defined in \eqref{egf2}.
\end{Corollary}

Next, the $m=0$ case of \eqref{transdef}, for $\im(t_i),\im(s_i)>0$, and $Z=0$, $Y=0$, gives the following result:

\begin{Corollary}\label{cor.mn0}
\begin{gather}\label{Anebsi}
I^0_{A_n}(0,0\,|\,\vect,\veca;\vecs,\vecb)=\prod_{j=0}^{n+1}\Gamma(T-t_j,A-a_j) \Gamma(S-s_j,B-b_j)\prod_{i=0}^{n+1}\Gamma(t_i+s_j,a_i+b_j).
\end{gather}
\end{Corollary}

\begin{Remark}
For $r=1$, equation \eqref{Anebsi} is equivalent to an elliptic beta integral identity associated to the $A_n$ root system \cite{RainsT,SpiridonovTheta,SpiridonovShortProofs}.
\end{Remark}

By specializing further to the case $m=0$, $n=1$ in Theorem \ref{mainthm}, we obtain the following sum/integral analogue of elliptic beta integral of Spiridonov:

\begin{Corollary}[\cite{Kels:2015bda}]\label{cor.Kels}
For complex variables $\sigma$, $\tau$, $t_i$, and integer variables $a_i$, $i=1,2,\ldots,6$, satisfying
\begin{gather*}
\im(\sigma),\;\im(\tau),\;\im(t_i)>0,\qquad\sum_{i=1}^6t_i\equiv\sigma+\tau \ (\text{\rm mod }2r),\qquad\sum_{i=1}^6a_i\equiv 0 \ (\text{\rm mod }r),
\end{gather*}
the following identity holds
\begin{gather}\label{ebsidef}
\frac{\lambda}{2}\sum^{r-1}_{y=0}\int_0^1{\rm d}z \frac{\prod\limits_{i=1}^6\Gamma(t_i+z,a_i+y) \Gamma(t_i-z,a_i-y)}{\Gamma(2z,2y)\Gamma(-2z,-2y)}=\prod_{1\leq i<j\leq6}\Gamma(t_i+t_j,a_i+a_j).
\end{gather}
\end{Corollary}

\begin{Remark}
For $r=1$, equation \eqref{ebsidef} is equivalent to the elliptic beta integral \cite{SpiridonovEBF}.
\end{Remark}

Note that there is a symmetry of the integrand in \eqref{ebsidef} under $z\rightarrow -z$, $y\rightarrow r-y$, which may be used to truncate the sum to values $0\leq y\leq \floor{r/2}$ \cite{GahramanovKels}.

The hyperbolic limit of \eqref{ebsidef} was also recently studied \cite{GahramanovKels} in connection with two-di\-men\-sio\-nal integrable lattice models of statistical mechanics, where it generalises the Faddeev--Volkov model \cite{Faddeev:1993pe,Volkov:1992uv}, and in connection with supersymmetric gauge theory, where it describes duality of three-dimensional $\mathcal{N} = 2$ theories on squashed lens spaces. As expected, the hyperbolic limit of~\eqref{ebsidef} provides a sum/integral generalisation of the hyperbolic beta integral~\cite{STOKMAN2005119}.

Finally we note that in a previous paper \cite{GahramanovKels}, it was shown that the formula \eqref{ebsidef} is also satisf\/ied when the normalisation function \eqref{r2def}, is def\/ined as
\begin{gather}\label{gennorm}
R_2(z,m;\sigma,\tau):=R(z+m\sigma;\zeta\sigma,\sigma+\tau)+R(z+(\zeta-m)\tau;\zeta\tau,\sigma+\tau),
\end{gather}
where $\zeta$ is a non-zero integer. The case $\zeta=r$ corresponds to the normalisation \eqref{r2def}, while the case $\zeta=1$ corresponds to the normalisation of the lens elliptic gamma function used in \cite{rarified}. We note that the main result in Theorem \ref{mainthm} (and also Theorem \ref{secondthm}) is also satisf\/ied when the normalisation function is chosen as \eqref{gennorm}, which can be checked by explicitly expanding both sides of \eqref{transdef}, and seeing that there is no dependence on $\zeta$. However the properties of the lens elliptic gamma function given in Section \ref{sec:defs}, are only true for the case of $\zeta=r$. Particularly, in this paper we always consider $\zeta=r$ unless explicitly stated otherwise.

\subsection{Proof of Theorem \ref{mainthm}}

The proof of Theorem \ref{mainthm} in the $r=1$ case \cite{RainsT} generalises to the $r>1$ case considered here. This involves f\/irst proving a special case of Theorem \ref{mainthm}, using a Frobenius type determinant formula for the lens theta functions \eqref{lthtdef}, and then taking limits of this special case, which are then used to show that a dense set of cases hold for the general transformation \eqref{transdef}.

First, we have the following determinant formula which follows directly from Lemma~4.3 of~\cite{RainsT}:
\begin{Lemma}\label{lem.Frobenius}
\begin{gather}
\det_{0\leq i,j<n} \left(\frac{\theta_k(t+x_i+w_j,c_i+d_j)}{\theta_k(t,0)\,\theta_k(x_i+w_j,c_i+d_j)}\right) = \frac{\theta_k(t+X+W,C+D)}{\theta_k(t,0)\prod\limits_{i,j=0}^{n-1}\theta_k(x_i+w_j,c_i+d_j)} \nonumber\\
\qquad{} \times \prod_{0\leq i<j<n} \EXP^{2\pi\ii(x_j+w_j-c_j-d_j)/r} \theta_k(x_i-x_j,c_i-c_j)\theta_k(w_i-w_j,d_i-d_j),\label{detident}
\end{gather}
where $\theta_{k=1,2}(z,m)$ represents either of $\theta_1(z,m)$ or $\theta_2(z,m)$ in \eqref{lthtdef}, and
\begin{gather*}
X=\sum_{i=0}^{n-1}x_i,\qquad W=\sum_{i=0}^{n-1}w_i,\qquad C=\sum_{i=0}^{n-1}c_i,\qquad D=\sum_{i=0}^{n-1}d_i.
\end{gather*}
\end{Lemma}

Note that both sides of the equation \eqref{detident} are periodic in both the complex and integer variables respectively, i.e., they each are invariant under the shifts $c_i\rightarrow c_i+k_c r$, $d_i\rightarrow d_i+k_d r$, $x_i\rightarrow x_i+k_x$, $w_j\rightarrow w_j+k_w$, for integers $k_c$, $k_d$, $k_x$, $k_w$.

The determinant \eqref{detident} will be used to prove the following special case of \eqref{transdef}.

\begin{Lemma}
Theorem {\rm \ref{mainthm}} holds for the case $m\to n-1, n\to n-1$, with the following choice of variables:
\begin{alignat}{5}
& t_{i}=x_i,\qquad && t_{n+i}=\tau-w_i,\qquad&& s_{i}=\sigma-x_i,\qquad && s_{n+i}=w_i,&\nonumber \\
& a_i=c_i,\qquad && a_{n+i}=1-d_i,\qquad && b_i=-1-c_i,\qquad && b_{n+i}=d_i,& \label{specialvars}
\end{alignat}
where $i=0,1,\ldots,n-1$.
\end{Lemma}

\begin{proof}
Consider the following univariate sum/integral related to \eqref{detident}
\begin{gather}\label{unisumintdet}
\sum_{y=0}^{r-1}\int {\rm d}z \frac{\theta_1(s+w-z,d-y)}{\theta_1(s,0) \theta_1(w-z,d-y)} \frac{\theta_2(t+x+z,c+y)}{\theta_2(t,0) \theta_2(x+z,c+y)}.
\end{gather}
This sum/integral is invariant upon exchanging $x\leftrightarrow w$, and $c\leftrightarrow d$, which follows from the change of integration and summation variables $z=z-x+w$, and $y=y-c+d$, respectively.

It then follows that a determinant of instances of \eqref{unisumintdet}
\begin{gather}\label{det1}
\det_{0\leq i,j<n}\left(\sum_{y=0}^{r-1}\int {\rm d}z \frac{\theta_1(s+w_j-z,d_j-y)}{\theta_1(s,0) \theta_1(w_j-z,d_j-y)}\frac{\theta_2(t+x_i+z,c_i+y)}{\theta_2(t,0) \theta_2(x_i+z,c_i+y)}\right),
\end{gather}
is invariant under the exchange
\begin{gather}
x_j\leftrightarrow w_j , \qquad c_j\leftrightarrow d_j .\label{symmetry}
\end{gather}

Since the row and column indices in \eqref{det1} are not coupled, it may be written in terms of a~multivariate sum/integral of a product of determinants of the form \eqref{detident}, as
\begin{gather*}
 n!\det_{0\leq i,j<n}\left(\sum_{y=0}^{r-1}\int {\rm d}z \frac{\theta_1(s+w_j-z,d_j-y)}{\theta_1(s,0) \theta_1(w_j-z,d_j-y)}\frac{\theta_2(t+x_i+z,c_i+y)}{\theta_2(t,0) \theta_2(x_i+z,c_i+y)}\right) \\
 \qquad{} =\sum_{y_0,\ldots,y_{n-1}=0}^{r-1}\int \det_{0\leq i,j<n}\left(\frac{\theta_1(s+w_i-z_j,d_i-y_j)}{\theta_1(s,0)\theta_1(w_i-z_j,d_i-y_j)}\right) \\
\qquad\quad{} \times\det_{0\leq i,j<n} \left(\frac{\theta_2(t+x_i+z_j,c_i+y_j)}{\theta_2(t,0) \theta_2(x_i+z_j,c_i+y_j)}\right) \prod_{i=0}^{n-1}{\rm d}z_i \\
\qquad{} = \prod_{0\leq i<j<n} \EXP^{2\pi\ii (x_j+w_j-c_j-d_j)/r} \sum_{y_0,\ldots,y_{n-1}=0}^{r-1}\int \tilde{\Delta}(\vecz,\vecy;\vecx,\vecc;\vecw,\vecd) \prod_{i=0}^{n-1}{\rm d}z_i .
\end{gather*}
The integrand in the last line is
\begin{gather*}
\tilde{\Delta}(\vecz,\vecy;\vecx,\vecc;\vecw,\vecd)
=\frac{\theta_1(s + W - Z,D - Y) \theta_2(t + X + Z,C + Y)}{\theta_1(s,0) \theta_2(t,0)} \Delta(\vecz,\vecy;\vecx,\vecc;\vecw,\vecd),
\end{gather*}
where
\begin{gather*}
 \Delta(\vecz,\vecy;\vecx,\vecc;\vecw,\vecd) \nonumber\\
 =\frac{\prod\limits_{0\leq i<j<n}\theta_1(w_i-w_j,d_i-d_j) \theta_2(x_i-x_j,c_i-c_j) \theta_1(z_i-z_j,y_i-y_j) \theta_2(z_j-z_i,y_j-y_i)}{\prod\limits_{i,j=0}^{n-1}\theta_1(w_j-z_i,d_j-y_i) \theta_2(x_j+z_i,c_j+y_i)},
\end{gather*}
and
\begin{gather*}
Z=\sum_{i=0}^{n-1}z_i,\qquad Y=\sum_{i=0}^{n-1}y_i.
\end{gather*}

The symmetry under \eqref{symmetry} implies the following equality
\begin{gather*}
\sum_{y_1,\ldots,y_{n-1}=0}^{r-1}\int \tilde{\Delta}(\vecz,\vecy;\vecx,\vecc;\vecw,\vecd)\prod_{i=0}^{n-1}{\rm d}z_i = \sum_{y_1,\ldots,y_{n-1}=0}^{r-1}\int \tilde{\Delta}(\vecz,\vecy;\vecw,\vecd;\vecx,\vecc)\prod_{i=0}^{n-1}{\rm d}z_i.
\end{gather*}
Substituting $t\rightarrow t+rk\sigma$ in this relation gives
\begin{gather*}
\sum_{y_1,\ldots,y_{n-1}=0}^{r-1}\int \frac{\tilde{\Delta}(\vecz,\vecy;\vecx,\vecc;\vecw,\vecd)}{\EXP^{2\pi\ii(X+Z)k}}\prod_{i=0}^{n-1}{\rm d}z_i = \sum_{y_1,\ldots,y_{n-1}=0}^{r-1}\int \frac{\tilde{\Delta}(\vecz,\vecy;\vecw,\vecd;\vecx,\vecc)}{\EXP^{2\pi\ii(W+Z)k}}\prod_{i=0}^{n-1}{\rm d}z_i,
\end{gather*}
for integers $k$. This in turn means that
\begin{gather*}
\sum_{y_1,\ldots,y_{n-1}=0}^{r-1}\int f(\EXP^{2\pi\ii(X+Z)}) \tilde{\Delta}(\vecz,\vecy;\vecx,\vecc;\vecw,\vecd)\prod_{i=0}^{n-1} {\rm d}z_i \\
\qquad{} = \sum_{y_1,\ldots,y_{n-1}=0}^{r-1}\int  f(\EXP^{2\pi\ii(W+Z)}) \tilde{\Delta}(\vecz,\vecy;\vecw,\vecd;\vecx,\vecc)\prod_{i=0}^{n-1} {\rm d}z_i,
\end{gather*}
for any function $f(z)$ holomorphic in a neighbourhood of the contour. Writing this in terms of integration over $Z$, and summation over $Y$, it is seen that the following equality must hold
\begin{gather*}
\sum_{\substack{y_0,\ldots,y_{n-2}=0 \\ \sum\limits_{i=0}^{n-1} y_i=Y}}^{r-1}
\int_{\;\sum\limits_{i=0}^{n-1} z_i=Z} \Delta(\vecz,\vecy;\vecx,\vecc;\vecw,\vecd) \prod_{i=0}^{n-2}{\rm d}z_i\\
\qquad{} =
\sum_{\substack{y_0,\ldots,y_{n-2}=0 \\ \sum\limits_{i=0}^{n-1} y_i=Y+C-D}}^{r-1}\int_{\;\sum\limits_{i=0}^{n-1} z_i=Z+X-W}
\Delta(\vecz,\vecy;\vecw,\vecd;\vecx,\vecc) \prod_{i=0}^{n-2}{\rm d}z_i.
\end{gather*}
This is equivalent to the transformation \eqref{transdef} with the variables \eqref{specialvars}.
\end{proof}

The above special case \eqref{specialvars} will be used, along with the following general limit of the $A_n$ sum/integral \eqref{AnIntDef}.

\begin{Lemma}
For $a_0=-b_0$, the limit $t_0\rightarrow-s_0$ of \eqref{AnIntDef} is given by
\begin{gather}
\lim_{t_0\rightarrow -s_0}\frac{\left.I^{m}_{A_n}(Z,Y\,|\,\vect,\veca;\vecs,\vecb)\right|_{a_0=-b_0}}{\Gamma(t_0+s_0,0)
\prod\limits_{i=1}^{m+n+1}\Gamma(t_0+s_i,-b_0+b_i)\Gamma(t_i+s_0,a_i+b_0)}\nonumber\\
\qquad{} =I^{m}_{A_{n-1}}(Z+s_0,Y+b_0\,|\,\bar{\vect},\bar{\veca};\bar{\vecs},\bar{\vecb}),\label{anlimit}
\end{gather}
where $\vect$, $\veca$, $\vecs$, $\vecb$ are as defined in \eqref{anvariables}, and
\begin{alignat*}{3}
& \bar{\vect} = (t_1,t_2,\ldots,t_{m+n+1}),\qquad&& \bar{\veca} = (a_1,a_2,\ldots,a_{m+n+1}),& \nonumber\\
& \bar{\vecs} = (s_1,s_2,\ldots,s_{m+n+1}),\qquad&& \bar{\vecb} = (b_1,b_2,\ldots,b_{m+n+1}).&
\end{alignat*}
\end{Lemma}

Note here that the value some of the other variables in $\vect$, $\veca$, $\vecs$, $\vecb$, should also be changed, if the balancing condition~\eqref{balancing} is to be satisf\/ied both before and after taking the limit \eqref{anlimit}.

\begin{proof}To take the limit \eqref{anlimit}, the contour will need to be deformed across the poles at
\begin{gather*}
z_i=-t_0\ (\mbox{mod }2r),\qquad y_i=-a_0 \ (\mbox{mod }r)=b_0 \ (\mbox{mod }r),\qquad i=0,1,\ldots,n,
\end{gather*}
for $i=0,1,\ldots,n$. Then since in this limit the numerator remains f\/inite and the denominator has a divergent factor, $\lim\limits_{t_0\rightarrow-s_0}\Gamma(t_0+s_0,0)$, the only non-zero contribution in the limit \eqref{anlimit} comes from the residues calculated at the above poles. By symmetry, the residue of each $z_i$ at the poles must contribute the same value to \eqref{anlimit}, resulting in a factor $n+1$. In the limit, the terms in the denominator of \eqref{anlimit}, cancel the required terms in $I^m_{A_n}$ to give the integrand of $I^m_{A_{n-1}}$. The remaining factors come from using $\lim\limits_{z_0\rightarrow-t_0}(t_0+z_0)\Gamma(t_0+z_0,0)=\ii/(2\pi\lambda)$, in calculating the residue.
\end{proof}

We can now prove Theorem \ref{mainthm}.

\begin{proof}[Proof of Theorem \ref{mainthm}]
Note that the special case of the transformation in \eqref{specialvars} has the sets of $2n$ variables $t_i$, $s_i$, $a_i$, $b_i$ (subject to balancing condition \eqref{balancing}) parameterised in terms of the sets of $n$ independent variables $x_i$, $w_i$, $c_i$, $d_i$. Similarly, in the following, the general sets of variables $t_i$, $s_i$, $a_i$, $b_i$ in~\eqref{AnIntDef} will be parameterised by the $m+n+2$ pairs of sums of variables, as
\begin{gather}\label{alphabetadef}
\{\alpha_i,\beta_i\}=\{t_i+s_i,a_i+b_i\},\qquad i=0,1,\ldots,m+n+1.
\end{gather}
Let also $\mathcal{C}_{mn}$ denote the set of points $(\{\alpha_0,\beta_0\},\{\alpha_1,\beta_1\},\ldots,\{\alpha_{m+n+1},\beta_{m+n+1}\})$, for which the transformation \eqref{transdef} holds. For example, in this notation, the special case \eqref{specialvars} derived above corresponds to $(\{\sigma,-1\},\{\sigma,-1\},\ldots,\{\sigma,-1\},\{\tau,1\},\{\tau,1\},\ldots,\{\tau,1\})\in \mathcal{C}_{nn}$.

Consider f\/irst a limit $t_0\rightarrow -s_1$, for $a_0=-b_1$, on both sides of the general transformation~\eqref{transdef}, for some point $(\{\alpha_0,\beta_0\},\{\alpha_1,\beta_1\},\ldots,\{\alpha_{m+n+1},\beta_{m+n+1}\})\in \mathcal{C}_{mn}$. The limit of the right hand side requires no contour deformation, and has a simple cancellation of all factors of lens elliptic gamma functions appearing in the integrand, which contain any of the $t_0$, $a_0$, or $s_1$, $b_1$, variables. This produces a sum/integral of the type $I^{n-1}_{A_m}$. The limit of the left hand side follows from the above limit~\eqref{anlimit}, resulting in a sum/integral of the type $I_{A_{n-1}}^m$. Overall, in terms of~\eqref{alphabetadef} this limit produces the transformation corresponding to the point $(\{\alpha_0+\alpha_1,\beta_0+\beta_1\},\{\alpha_2,\beta_2\},\ldots,\{\alpha_{m+n+1},\beta_{m+n+1}\})\in \mathcal{C}_{m(n-1)}$.

Next consider the limit $t_0\rightarrow -s_1+\sigma+\tau$, for $a_0=-b_1$, again on both sides of the transformation \eqref{transdef}, for $(\{\alpha_0,\beta_0\},\{\alpha_1,\beta_1\},\ldots,\{\alpha_{m+n+1},\beta_{m+n+1}\})\in \mathcal{C}_{mn}$. This case is rather similar to the case in the previous paragraph, but here the left hand side of \eqref{transdef} now involves a~simple cancellation, resulting in $I^{m-1}_{A_n}$, and the right hand side gives the limit via~\eqref{anlimit}, resulting in $I^n_{A_{m-1}}$. Overall, this limit results in the transformation that corresponds to the point $(\{\alpha_0+\alpha_1 -\sigma-\tau,\beta_0+\beta_1\},\{\alpha_2,\beta_2\},\ldots,\{\alpha_{m+n+1},\beta_{m+n+1}\})\in \mathcal{C}_{(m-1)n}$.

Now starting at the point $(\{\sigma,-1\},\{\sigma,-1\},\ldots,\{\sigma,-1\},\{\tau,1\},\{\tau,1\},\ldots,\{\tau,1\})\in \mathcal{C}_{nn}$, corresponding to the special case~\eqref{specialvars}, the above two limits may be used to show that
\begin{gather}
(\{2\sigma,-2\},\{\sigma,-1\},\ldots,\{\sigma,-1\},\{\tau,1\},\{\tau,1\},\ldots,\{\tau,1\})\in \mathcal{C}_{n(n-1)}, \nonumber\\
(\{\sigma-\tau,-2\},\{\sigma,-1\},\ldots,\{\sigma,-1\},\{\tau,1\},\{\tau,1\},\ldots,\{\tau,1\})\in \mathcal{C}_{(n-1)n}, \nonumber\\
(\{\sigma,-1\},\{\sigma,-1\},\ldots,\{\sigma,-1\},\{2\tau,2\},\{\tau,1\},\ldots,\{\tau,1\})\in \mathcal{C}_{n(n-1)}, \nonumber\\
(\{\sigma,-1\},\{\sigma,-1\},\ldots,\{\sigma,-1\},\{\tau-\sigma,2\},\{\tau,1\},\ldots,\{\tau,1\})\in \mathcal{C}_{(n-1)n}.\label{contrelations2}
\end{gather}
Iterating the above limits a further $r-1$ times, gives
\begin{gather}
(\{(r+1)\sigma,-1\},\{\sigma,-1\},\ldots,\{\sigma,-1\},\{\tau,1\},\{\tau,1\},\ldots,\{\tau,1\})\in \mathcal{C}_{n(n-r)}, \nonumber\\
(\{\sigma-r\tau,-1\},\{\sigma,-1\},\ldots,\{\sigma,-1\},\{\tau,1\},\{\tau,1\},\ldots,\{\tau,1\})\in \mathcal{C}_{(n-r)n}, \nonumber\\
(\{\sigma,-1\},\{\sigma,-1\},\ldots,\{\sigma,-1\},\{(r+1)\tau,1\},\{\tau,1\},\ldots,\{\tau,1\})\in \mathcal{C}_{n(n-r)}, \nonumber\\
(\{\sigma,-1\},\{\sigma,-1\},\ldots,\{\sigma,-1\},\{\tau-r\sigma,1\},\{\tau,1\},\ldots,\{\tau,1\})\in \mathcal{C}_{(n-r)n}.\label{contrelations}
\end{gather}
In iterating the relations \eqref{contrelations2} $r$ times, the integer component $\beta_i$ of the $\{\tau,1\}$ or $\{\sigma,-1\}$, cycle through $r$ dif\/ferent values $+1+k$, and $-1-k$, respectively, for $k=0,1,\ldots,r-1$.

Now consider starting from an arbitrary value of the $n$ integer components $\beta_i$, which may be arrived at by using \eqref{contrelations2} up to $r$ times for each pair $\{\alpha_i,\beta_i\}$. The relations \eqref{contrelations} can then be repeatedly iterated to form arbitrary combinations (that are subject to the balancing condition~\eqref{balancing}) of the form $\alpha_i=j_1r\tau-j_2r\sigma+k_1\tau-k_2\sigma$, or $\alpha_i=j_1r\sigma-j_2r\tau+k_1\sigma-k_2\tau$, for some integers $j_1,j_2=0,1,\ldots$, and $k_1,k_2=0,1,\ldots, r-1$. Taking into account the $2r$-periodicity of the $\alpha_i$, as $n\rightarrow\infty$ this gives a dense set of points for the $\alpha_i$ in $C_{mn}$ (for any choice of the $\beta_i$), and thus Theorem~\ref{mainthm} holds in general.
\end{proof}

\section[The $BC_n\leftrightarrow BC_m$ transformation]{The $\boldsymbol{BC_n\leftrightarrow BC_m}$ transformation}\label{sec:bcntrans}

\subsection{Main theorem}
Let us introduce complex variables $\sigma$, $\tau$, $t_i$, and integer variables $a_i$, for $i=0,1,\ldots,2m+2n+3$, satisfying
\begin{gather}
\im(\sigma),\im(\tau)>0, \nonumber\\
\sum_{i=0}^{2m+2n+3} t_i\equiv(m+1)(\sigma+\tau)\ (\textrm{mod }2r), \qquad \sum_{i=0}^{2m+2n+3} a_i\equiv 0 \ (\textrm{mod }r).\label{balancing2}
\end{gather}
In terms of these variables, $I^m_{BC_n}(\vect,\veca)$ is def\/ined to be the following elliptic hypergeometric sum/integral
\begin{gather}\label{BCnIntDef}
I^m_{BC_n}(\vect,\veca):=\frac{\lambda^n}{2^n\,n!}\,\sum_{y_1,\ldots,y_n=0}^{r-1}\,\int_{C^n}\Delta^m_{BC_n}(\vecz,\vecy;\vect,\veca)\prod_{i=1}^n{\rm d}z_i,
\end{gather}
where
\begin{alignat}{3}
& \vecz = (z_1,z_2,\ldots,z_{n}),\qquad&&\vect = (t_0,t_1,\ldots,t_{2m+2n+3}),& \nonumber\\
& \vecy = (y_1,y_2,\ldots,y_{n}),\qquad&&\veca = (a_0,a_1,\ldots,a_{2m+2n+3}),& \label{bcnvariables}
\end{alignat}
and
\begin{gather}\label{integrand2}
\Delta^m_{BC_n}(\vecz,\vecy;\vect,\veca):=\frac{\prod\limits_{i=1}^n\prod\limits_{j=0}^{2m+2n+3}\Gamma(t_j+z_i,a_j+y_i) \Gamma(t_j-z_i,a_j-y_i)}{\prod\limits_{i=1}^n\Gamma(\pm 2z_i,\pm 2y_i)\prod\limits_{1\leq i<j\leq n}\Gamma(\pm z_i\pm z_j,\pm y_i\pm y_j)}.
\end{gather}
Here the compact notation for the lens elliptic gamma function is now used, where for example
\begin{gather}
\Gamma(\pm z_1,\pm y_1)=\Gamma(z_1,y_1)\Gamma(-z_1,-y_1), \nonumber\\
\Gamma(\pm z_1\pm z_2,\pm y_1 \pm y_2)= \Gamma(z_1+z_2,y_1+y_2) \Gamma(z_1-z_2,y_1-y_2) \nonumber\\
 \phantom{\Gamma(\pm z_1\pm z_2,\pm y_1 \pm y_2)=}{} \times\Gamma(-z_1+z_2,-y_1+y_2) \Gamma(-z_1-z_2,-y_1-y_2),\label{compnot}
\end{gather}
for complex numbers $z_1$, $z_2$, and integers $y_1$, $y_2$, respectively.

For the values $\im(t_i)>0$, the contour in \eqref{BCnIntDef} may be chosen to be $C^n$, where $C=[0,1]$. Otherwise, $C$ is a contour that connects the two points $\ii k$, and $1+\ii k$, where $k$ is a real number, and $C$ also separates all points
\begin{gather*}
\{t_j+(\sigma+\tau)l_1+\sigma(rl_2+(a_j-y_i)\mbox{ mod }r)+n, \\
 \qquad t_j+(\sigma+\tau)l_1+\tau(r(l_2+1)-(a_j-y_i)\mbox{ mod }r)+n\},
\end{gather*}
in the strip $0\leq \re(z)\leq 1$, from the points
\begin{gather*}
\{-t_j-(\sigma+\tau)l_1-\sigma(rl_2+(a_j+y_i)\mbox{ mod }r)+n, \\
\qquad {-}t_j-(\sigma+\tau)l_1-\tau(r(l_2+1)-(a_i+y_i)\mbox{ mod }r)+n\},
\end{gather*}
in the same strip, where $j=0,\ldots,2m+2n+3$, $l_1,l_2=0,1,\ldots$, and $n\in\mathbb{Z}$.

The integrand \eqref{integrand2} is periodic in the summation and integration variables $y_i$, and $z_i$, respectively, i.e., $\Delta^m_{BC_n}$ is invariant under either $z_i\rightarrow z_i+k_i$, or $y_i\rightarrow y_i+rk_i$, $i=1,2,\ldots,n$, for any $k_i\in\mathbb{Z}$. The periodicity of \eqref{integrand2} in the $z_i$ follows from the balancing condition $\sum\limits_{i=0}^{2m+2n+3}a_i\equiv 0$ $(\textrm{mod }r)$. The integrand is also $2r$-periodic in the complex variables $t_i$.

Note also the particular case $n=0$ of \eqref{BCnIntDef} gives
\begin{gather*}
I^m_{BC_0}(\vect,\veca)=1,
\end{gather*}
while the case $n=1$ of \eqref{BCnIntDef}, is equivalent to the $A_n$ integral \eqref{AnIntDef}
\begin{gather}\label{bcnan}
I^m_{BC_1}(\vect,\veca)=I^m_{A_1}(\vect_1,\veca_1;\vect_2,\veca_2),
\end{gather}
where a choice of the variables in \eqref{bcnan} is
\begin{alignat*}{3}
& \vect_1 = (t_0,t_1,\ldots,t_{m+n+1}),\qquad&&\vect_2 = (t_{m+n+2},t_{m+n+3},\ldots,t_{2m+2n+3}), & \\
& \veca_1 = (a_0,a_1,\ldots,a_{m+n+1}),\qquad&&\veca_2 =(a_{m+n+2},a_{m+n+3},\ldots,a_{2m+2n+3}).&
\end{alignat*}

The $BC_n$ sum/integral \eqref{BCnIntDef} satisf\/ies the following transformation formula which is the main result of this section.

\begin{Theorem}\label{secondthm}
The sum/integral \eqref{BCnIntDef}, under the balancing condition \eqref{balancing2}, satisfies
\begin{gather}\label{transdef2}
I^m_{BC_n}(\vect,\veca)=I^n_{BC_m}(\tvect,\tveca)\prod_{0\leq i<j\leq 2m+2n+3}\Gamma(t_i+t_j,a_i+a_j),
\end{gather}
where
\begin{gather*}
\tvect=\frac{\sigma+\tau}{2}-\vect,\qquad\tveca=-\veca.
\end{gather*}
\end{Theorem}

\subsection{Colloraries}

The $r=1$ case of Theorem \ref{secondthm} is equivalent to the $BC_n\leftrightarrow BC_m$ elliptic hypergeometric integral transformations given by Rains~\cite{RainsT}.

\begin{Corollary}[{\cite{RainsT}}]
For $r=1$, \eqref{transdef2} is
\begin{gather*}
\frac{\lambda^n}{2^n n!}\int_{C^n} \Delta^m_{BC_n}(\vecz,\vect)\prod_{i=1}^n{\rm d}z_i=\frac{\lambda^m}{2^m m!}\int_{C^m} \Delta^n_{BC_m}(\vecz,\tvect)\prod_{i=1}^m{\rm d}z_i \prod_{0\leq i<j\leq 2m+2n+3} \Gamma_1(t_i+t_j),
\end{gather*}
where
\begin{gather*}
\Delta^m_{BC_n}(\vecz,\vect)=\frac{\prod\limits_{i=1}^n\prod\limits_{j=0}^{2m+2n+3}\Gamma_1(t_j+z_i) \Gamma_1(t_j-z_i)}{\prod\limits_{i=1}^n\Gamma_1(\pm 2z_i)\prod\limits_{1\leq i<j\leq n}\Gamma_1(\pm z_i\pm z_j)},
\end{gather*}
and $\Gamma_1(z)=\Gamma_1(z;\sigma,\tau)$ is the usual elliptic gamma function defined in \eqref{egf2}.
\end{Corollary}

The $m=0$ case of Theorem \ref{secondthm} is equivalent to a $BC_n$ elliptic hypergeometric sum/integral identity proven by Spiridonov~\cite{rarified} (named there ``raref\/ied elliptic hypergeometric integral'').

\begin{Corollary}[\cite{rarified}] For $m=0$, \eqref{transdef2} is
\begin{gather}\label{bcnebsi}
\frac{\lambda^n}{2^n n!}\sum_{y_1,\ldots,y_n=0}^{r-1}\int_{C^n}\Delta^m_{BC_n}(\vecz,\vecy;\vect,\veca)\prod_{i=1}^n{\rm d}z_i= \prod_{0\leq i<j\leq 2n+3} \Gamma(t_i+t_j,a_i+a_j).
\end{gather}
\end{Corollary}

\begin{Remark}
The $n=1$ case of \eqref{bcnebsi} is equivalent to the elliptic beta/sum integral \eqref{ebsidef}.
\end{Remark}

As was mentioned in the previous section, Theorem \ref{mainthm}, and Theorem \ref{secondthm}, remain satisf\/ied when the normalisation \eqref{r2def} is replaced with \eqref{gennorm}. The normalisation of \eqref{bcnebsi} in \cite{rarified} corresponds to $\zeta=1$ in \eqref{gennorm}.

\subsection{Proof of Theorem \ref{secondthm}}

The proof of Theorem \ref{secondthm} basically follows the same idea as in the proof of Theorem \ref{mainthm}, with minor dif\/ferences. In fact a special case of Theorem \ref{mainthm} is f\/irst used to prove the following Lemma.

\begin{Lemma}\label{bcnn1}
Theorem {\rm \ref{secondthm}} holds for $m=1$, $n=1$, with the following choice of the variables
\begin{alignat}{5}
& t_{0} = x_0,\qquad && t_{1} = \tau-x_0,\qquad && t_{2} = x_1,\qquad && t_{3} = \tau-x_1,& \nonumber \\
& t_{4} = w_0,\qquad && t_{5} = \sigma-w_0,\qquad && t_{6} = w_1,\qquad && t_{7} = \sigma-w_1, & \nonumber \\
&a_{0} = c_0,\qquad&& a_{1}=1-c_0,\qquad&& a_{2}=c_1,\qquad&& a_{3}=1-c_1, & \nonumber \\
& a_{4}=d_0,\qquad&& a_{5}=-1-d_0,\qquad&& a_{6}=d_1,\qquad&& a_{7}=-1-d_1.&\label{specialcase2n1}
\end{alignat}
\end{Lemma}

\begin{proof}
The case \eqref{specialcase2n1} of the transformation \eqref{transdef2}, is explicitly given by
\begin{gather}
 \sum_{y=0}^{r-1}\int_C\frac{\theta_1(2z,2y) \theta_2(-2z,-2y)}{\prod\limits_{j\in\{0,1\}}\theta_1(x_j\pm z,c_j\pm y) \theta_2(w_j\pm z,d_j\pm y)} {\rm d}z \nonumber\\
\qquad{} =\sum_{y=0}^{r-1}\int_C\frac{\theta_1(2z,2y) \theta_2(-2z,-2y)}{\prod\limits_{j\in\{0,1\}}\theta_2(\frac{\sigma+\tau}{2}-x_j\pm z,-c_j\pm y) \theta_1(\frac{\sigma+\tau}{2}-w_j\pm z,-d_j\pm y)}{\rm d}z \nonumber\\
\qquad\quad{}\times\frac{\theta_2(x_0-x_1,c_0-c_1) \theta_2(x_0+x_1-\tau,c_0+c_1-1)}{\theta_1(x_0-x_1,c_0-c_1) \theta_1(x_0+x_1-\tau,c_0+c_1-1)} \nonumber\\
\qquad\quad{} \times\frac{\theta_1(w_0-w_1,d_0-d_1) \theta_1(w_0+w_1-\sigma,d_0+d_1+1)}{\theta_2(w_0-w_1,d_0-d_1) \theta_2(w_0+w_1-\sigma,d_0+d_1+1)}.\label{n1sc}
\end{gather}

This is equivalent to an $m=n=1$ case of Theorem \ref{mainthm}, with the following choice of variables satisfying \eqref{balancing}
\begin{alignat*}{5}
& t_0 = x_0,\qquad &&t_1 = \tau-x_0,\qquad && t_2 = w_0,\qquad && t_3 = \sigma-w_0,& \\
&s_0=x_1,\qquad &&s_1 = \tau-x_1,\qquad &&s_2=w_1,\qquad &&s_3 =\sigma-w_1,& \\
&a_0=c_0,\qquad &&a_1 = +1-c_0,\qquad &&a_2 = d_0,\qquad &&a_3 = -1-d_0,& \\
&b_0=c_1,\qquad &&b_1 = +1-c_1,\qquad &&b_2 = d_1,\qquad &&b_3 = -1-d_1.&\tag*{\qed}
\end{alignat*}\renewcommand{\qed}{}
\end{proof}

Note that Lemma \ref{bcnn1} appears to be a special case of an $m=n=1$ identity given in \cite{rarified}, corresponding to the normalisation of the lens elliptic gamma function where $\zeta=1$ in \eqref{gennorm}.

To prove a more general case of Lemma \ref{bcnn1}, the special case of \eqref{n1sc} will be used, along with the following determinant identity, which follows directly from equation~(3.18) in~\cite{RainsT}.

\begin{Lemma}
\begin{gather}
\det_{1\leq i,j\leq n} \left( \frac{1}{\theta_k(x_i\pm w_j,c_i\pm d_j)} \right) = \prod_{1\leq i<j\leq n} \frac{\theta_k(x_i\pm x_j,c_i\pm c_j) \theta_k(w_i\pm w_j,d_i\pm d_j)}{\EXP^{-2\pi\ii(x_j-w_i-c_j+d_i+r/2)/r}} \nonumber\\
\hphantom{\det_{1\leq i,j\leq n} \left( \frac{1}{\theta_k(x_i\pm w_j,c_i\pm d_j)} \right) =}{} \times\prod_{i,j=1}^n\frac{1}{\theta_k(x_i\pm w_j,c_i\pm d_j)},\label{detident2}
\end{gather}
where $\theta_k(z,m)$ represents either of $\theta_1(z,m)$ or $\theta_2(z,m)$ in \eqref{lthtdef}.
\end{Lemma}

In \eqref{detident2}, the compact notation for the lens theta functions follows analogously to \eqref{compnot}.

Note that both sides of the equation \eqref{detident2} are periodic in both the complex and integer variables respectively, i.e., they each are invariant under the shifts $c_i\rightarrow c_i+k_1r$, $d_i\rightarrow d_i+k_2r$, $x_i\rightarrow x_i+k_3$, $w_j\rightarrow w_j+k_4$, for integers $k_1$, $k_2$, $k_3$, $k_4$.

The results \eqref{n1sc}, \eqref{detident2}, are used to prove the following special case of Theorem \ref{secondthm}.

\begin{Lemma} The transformation \eqref{transdef2} holds for $m=n$, with the following choice of variables
\begin{alignat}{5}
& t_{2i} = x_i,\qquad & & t_{2i+1} = \tau-x_i,\qquad&& t_{2n+2i+2} = w_i,\quad&& t_{2n+2i+3} = \sigma-w_i,&\nonumber \\
& a_{2i} = c_i,\qquad&& a_{2i+1} = 1-c_i,\qquad&& a_{2n+2i+2} = d_i,\qquad&& a_{2n+2i+3} = -1-d_i.&\label{specialvars2}
\end{alignat}
\end{Lemma}

\begin{proof}
Consider taking a determinant of particular instances of \eqref{n1sc}. Using \eqref{detident2}, the relevant determinant coming from the left hand side of \eqref{n1sc} may be written as
\begin{gather}
 n!\det_{1\leq i,j\leq n}\left(\sum_{y=0}^{r-1}\int_{C}\frac{\theta_1(2z,2y) \theta_2(-2z,-2y){\rm d}z}{\prod\limits_{k\in\{0,i\}}\theta_1(x_k\pm z,c_k\pm y) \prod\limits_{k\in\{0,j\}}\theta_2(w_k\pm z,d_k\pm y)}\right) \nonumber\\
\qquad {} = \sum_{y_1,\ldots,y_n=0}^{r-1}\int_{C^n}\det_{1\leq i,j\leq n} \left(\frac{\theta_1(2z_j,2y_j)}{\prod\limits_{k\in\{0,i\}}\theta_1(x_k\pm z_j,c_k\pm y_j)} \right)\nonumber\\
\qquad\quad{}\times \det_{1\leq i,j\leq n} \left(\frac{\theta_2(-2z_j,-2y_j)}{\prod\limits_{k\in\{0,i\}}\theta_2(w_k\pm z_j,d_k\pm y_j)} \right)\prod_{i=1}^n{\rm d}z_i \nonumber\\
 \qquad {} = \sum_{y_1,\ldots,y_n=0}^{r-1}\int_{C^n}\Delta(\vecz,\vecy;\vecx,\vecc;\vecw,\vecd)\prod_{i=1}^n{\rm d}z_i,\label{detlhs}
\end{gather}
where the integrand in the last line is
\begin{gather}
 \Delta(\vecz,\vecy;\vecx,\vecc;\vecw,\vecd)=\frac{\prod\limits_{i=1}^n\theta_1(2z_i,2y_i) \theta_2(-2z_i,-2y_i)}{\prod\limits_{i=1}^n\prod\limits_{j=0}^n\theta_1(x_j\pm z_i,c_j\pm y_i) \theta_2(w_j\pm z_i,d_j\pm y_i)}\label{specialintegrand2}\\
{} \times \prod_{1\leq i<j\leq n} \frac{\theta_1(x_i\pm x_j,c_i\pm c_j) \theta_2(w_i\pm w_j,d_i\pm d_j)}{\EXP^{-2\pi \ii(x_j+w_j-c_j-d_j)/r}} \theta_1(z_i\pm z_j,y_i\pm y_j) \theta_2(-z_i\pm z_j,-y_i\pm y_j).\nonumber
\end{gather}

Similarly, the corresponding determinant coming from the right hand side of \eqref{n1sc} is
\begin{gather}
 n! \det_{1\leq i,j\leq n} \left(\sum_{y=0}^{r-1}\int_{C}\frac{\theta_1(2z,2y) \theta_2(-2z,-2y){\rm d}z}{\prod\limits_{k\in\{0,j\}}\theta_1(\frac{\sigma+\tau}{2}-w_k\pm z,-d_k\pm y) \prod\limits_{k\in\{0,i\}}\theta_2(\frac{\sigma+\tau}{2}-x_k\pm z,-c_k\pm y)} \right. \nonumber\\
 \left.\qquad\quad{} \times\frac{\theta_2(x_0-x_i,c_0-c_i) \theta_2(x_0+x_i-\tau,c_0+c_i-1)}{\theta_1(x_0-x_i,c_0-c_i) \theta_1(x_0+x_i-\tau,c_0+c_i-1)}\right. \nonumber\\
 \left.\qquad\quad{} \times\frac{\theta_1(w_0-w_j,d_0-d_j) \theta_1(w_0+w_j-\sigma,d_0+d_j+1)}{\theta_2(w_0-w_j,d_0-d_j) \theta_2(w_0+w_j-\sigma,d_0+d_j+1)}\right) \nonumber\\
\qquad{} =\prod_{i=1}^n\frac{\theta_2(x_0-x_i,c_0-c_i) \theta_2(x_0+x_i-\tau,c_0+c_i-1)}{\theta_1(x_0-x_i,c_0-c_i) \theta_1(x_0+x_i-\tau,c_0+c_i-1)} \nonumber\\
 \qquad\quad{} \times\frac{\theta_1(w_0-w_i,d_0-d_i) \theta_1(w_0+w_i-\sigma,d_0+d_i+1)}{\theta_2(w_0-w_i,d_0-d_i) \theta_2(w_0+w_i-\sigma,d_0+d_i+1)} \nonumber\\
 \qquad\quad{}\times \sum_{y_1,\ldots,y_n=0}^{r-1} \int_{C^n} \Delta\left(\vecz,\vecy;\frac{\sigma+\tau}{2}-\vecw,-\vecd;\frac{\sigma+\tau}{2}-\vecx,-\vecc\right)\prod_{i=1}^n{\rm d}z_i,
\label{detrhs}\end{gather}
where $\Delta$ is def\/ined in \eqref{specialintegrand2}.

Since they were constructed from instances of \eqref{n1sc}, the expressions \eqref{detlhs} and \eqref{detrhs} must be equal, and this is exactly the transformation \eqref{transdef2} with variables \eqref{specialvars2}.
\end{proof}

Consider next the following limit of the $BC_n$ sum/integral \eqref{BCnIntDef}.

\begin{Lemma}
For $a_0=-a_1$, the limit $t_0\rightarrow-t_1$ of \eqref{BCnIntDef} is given by
\begin{gather}\label{bcnlimit}
\lim_{t_0\rightarrow -t_1}\frac{\left.I^{m}_{BC_n}(\vect,\veca)\right|_{a_0=-a_1}}{\Gamma(t_0+t_1,0) \!\!\prod\limits_{2\leq i\leq 2m+2n+3}\!\!\Gamma(t_0+t_i,-a_1+a_i) \Gamma(t_1+t_i,a_1+a_i)}=I^{m}_{BC_{n-1}}(\bar{\vect},\bar{\veca}), \!\!\!
\end{gather}
where $\vect, \veca$, are as given in \eqref{bcnvariables}, and
\begin{gather*}
\bar{\vect}=(t_2,t_3,\ldots,t_{2m+2n+3}),\qquad\bar{\veca}=(a_2,a_3,\ldots,a_{2m+2n+3}).
\end{gather*}
\end{Lemma}

Note here that the value some of the other variables in $\vect$, $\veca$, should also be changed, if the balancing condition \eqref{balancing2} is to be satisf\/ied both before and after taking the limit \eqref{bcnlimit}.

\begin{proof}
Similarly to \eqref{anlimit}, the contour needs be deformed to cross the poles at
\begin{alignat*}{3}
& z_i = -t_0\ (\mbox{mod }2r), \qquad && \text{for} \quad y_i=-a_0 \ (\mbox{mod }r)=+a_1 \ (\mbox{mod }r), & \\
& z_i = +t_0\ (\mbox{mod }2r),\qquad && \text{for} \quad y_i = +a_0 \ (\mbox{mod }r)=-a_1 \ (\mbox{mod }r),&
\end{alignat*}
for $i=0,1,\ldots,n$. The numerator remains f\/inite in the limit, and the only non-zero contribution to~\eqref{bcnlimit} comes from the residues at the above poles, which by symmetry each have the same value, resulting in a factor $2n$. At the poles, the terms in the denominator of \eqref{bcnlimit}, cancel the required terms in $I^m_{BC_n}$ to give the integrand of $I^m_{BC_{n-1}}$. The remaining factors come from using $\lim\limits_{z_0\rightarrow -t_0}(t_0+z_0)\Gamma(t_0+z_0,0)=i/(2\pi\lambda)$, and similarly for $z_0\rightarrow t_0$, in calculating the residues.
\end{proof}

The proof of Theorem \ref{secondthm}, now follows analogously to Theorem \ref{mainthm}, however with $\mathcal{C}_{mn}$ now consisting of the points of the type $(\{\alpha_0,\beta_0\},\{\alpha_1,\beta_1\},\ldots,\{\alpha_{m+n+1},\beta_{m+n+1}\})$, where
\begin{gather*}
\alpha_i=t_{2m+2n+2-2i}+t_{2m+2n+3-2i},\qquad \beta_i=a_{2m+2n+2-2i}+a_{2m+2n+3-2i},
\end{gather*}
where $i=0,1,\ldots, m+n+1$.

The limit \eqref{bcnlimit} implies that if $(\{\alpha_0,\beta_0\},\{\alpha_1,\beta_1\},\ldots,\{\alpha_{m+n+1},\beta_{m+n+1}\})\in \mathcal{C}_{mn}$, then
\begin{gather}
(\{\alpha_0+\alpha_1,\beta_0+\beta_1\},\{\alpha_2,\beta_2\},\ldots,\{\alpha_{m+n+1},\beta_{m+n+1}\})\in \mathcal{C}_{m(n-1)}, \nonumber\\
(\{\alpha_0+\alpha_1-(\sigma+\tau),\beta_0+\beta_1\},\{\alpha_2,\beta_2\},\ldots,\{\alpha_{m+n+1},\beta_{m+n+1}\})\in \mathcal{C}_{(m-1)n}.\label{abovelimits}
\end{gather}

Then starting from the special case \eqref{specialvars2}, corresponding to the point
\begin{gather*}
(\{\sigma,-1\},\{\sigma,-1\},\ldots,\{\sigma,-1\},\{\tau,1\},\{\tau,1\},\ldots,\{\tau,1\})\in \mathcal{C}_{nn},
\end{gather*}
the above limits \eqref{abovelimits} imply that analogous relations to \eqref{contrelations2}, \eqref{contrelations} hold for Theorem \ref{secondthm}. Then following the same argument of the $A_n$ case, this implies that for $n\rightarrow\infty$ we have a dense set of points for the $\alpha_i$ in $\mathcal{C}_{mn}$ (for any $\beta_i$), and thus Theorem~\ref{secondthm} holds in general.

\section{Application to supersymmetric gauge theories}\label{sec:SUSY}

We now come to some applications of the $A_n$ and $BC_n$ elliptic hypergeometric sum/integral transformation formulas, given in Theorems~\ref{mainthm} and \ref{secondthm} respectively.

The f\/irst application, discussed in this section, is in the area of the supersymmetric gauge theories: quantitative checks of the Seiberg duality \cite{Seiberg:1994pq} for supersymmetric QCD (SQCD) with gauge groups ${\rm SU}(N_c)$ and ${\rm Sp}(2N_c)$,
at the level of the lens index \cite{Benini:2011nc}. This generalises the similar considerations for the $r=1$ case \cite{Dolan:2008qi}, to more general cases with $r>1$.

\subsection[${\rm SU}(N_c)$]{$\boldsymbol{{\rm SU}(N_c)}$}

Let us f\/irst consider the gauge group ${\rm SU}(N_c)$. The Seiberg duality claims an equivalence between the following two theories, electric and magnetic.

{\bf Electric theory.}
The electric theory is the ${\rm SU}(N_c)$ SQCD with $N_f$ f\/lavors. In addition to an $\mathcal{N}=1$ vector multiplet $V$ associated with the ${\rm SU}(N_c)$ gauge group, we also have $\mathcal{N}=1$ chiral multiplets $q$ and $\bar{q}$, which are in the fundamental and anti-fundamental representation of~${\rm SU}(N_c)$ gauge symmetry, respectively.
This theory has no superpotential, $W_{\rm electric}=0$.

This theory has ${\rm SU}(N_f)_L\times {\rm SU}(N_f)_R \times {\rm U}(1)_B \times {\rm U}(1)_R$ non-anomalous f\/lavor symmetries,
where ${\rm U}(1)_R$ is the $R$-symmetry. The charge assignment of the f\/ields under the gauge/f\/lavor symmetries
is listed in Table~\ref{SUE}.

\begin{table}[htbp]\centering
\caption{Electric theory for ${\rm SU}(N_c)$ Seiberg duality.}\label{SUE}
\vspace{2mm}

\begin{tabular}{c||c|c|c|c|c}
 & ${\rm SU}(N_c)$ & ${\rm SU}(N_f)_L$ &${\rm SU}(N_f)_R$& ${\rm U}(1)_B$ & ${\rm U}(1)_R$ \\
 \hline
 \hline
$q$ & $\square$ & $\square$ & $\bm{1}$ &$1$ &$r_q=1-N_c/N_f$ \\
 \hline
$\bar{q}$ & $\overline{\square}$ & $\bm{1}$ & $\overline{\square}$ &$-1$ & $r_q=1-N_c/N_f$ \\
 \hline
 $V$ & adj.& $\bm{1}$ & $\bm{1}$&$0$ &$1$
\end{tabular}
\end{table}

To write down the lens index for this theory, let us prepare continuous/discrete fugacities for each f\/lavor symmetry. We have $(\vecz, \vecy)$ for ${\rm SU}(N_c)$ gauge symmetry, where both $\vecz$ and $\vecy$ are $N_c$-component vectors and the components of $\vecz$ ($\vecy$) are
complex parameters (integers taking values in $\mathbb{Z}_r$); these represent the discrete holonomies of the gauge f\/ields along the $S^1$ and the torsion cycle of $S^3/\mathbb{Z}_r$ of the geometry $S^1\times S^3/\mathbb{Z}_r$. Similarly, we have $(\bvect, \bveca)$ and $(\bvecs, \bvecb)$ for ${\rm SU}(N_f)_L$ and
${\rm SU}(N_f)_R$ f\/lavor symmetries. Since these correspond to $\rm SU$ symmetries (and not~$\rm U$), we have the constraints
\begin{gather*}
\sum_{i=1}^{N_c-1} z_i=\sum_{i=1}^{N_f-1} \bar{t}_i=\sum_{i=1}^{N_f-1} \bar{s}_i=0 , \qquad
\sum_{i=1}^{N_c-1} y_i \equiv \sum_{i=1}^{N_f-1} \bar{a}_i \equiv \sum_{i=1}^{N_f-1} \bar{b}_i\equiv 0\quad (\textrm{mod } r) .
\end{gather*}
We also need a pair $(B, n_B)$ for the ${\rm U}(1)_B$ symmetry, where again $B$ is a continuous parameter and $n_B$ is an integer. Finally, $R$-symmetry plays a~distinguished role and is associated with the pair $(R, 0)$, where $R$ is determined by $\sigma$ and $\tau$ as
\begin{gather*}
R=\frac{\tau+\sigma}{2} .
\end{gather*}
Note that integer-valued fugacity for $R$-symmetry is absent.

The lens index of this electric theory is then written as \cite{Benini:2011nc,Razamat:2013opa}
\begin{gather}
(\mathcal{I}_e)^{N_f}_{{\rm SU}(N_c)}(\vecz, \vecy; \bvect, \bveca; \bvecs, \bvecb; B, n_B)\nonumber\\
\qquad{}= \sum_{\substack{y_0,\ldots,y_{n-1}=0 \\ \sum\limits_{i=0}^{n}y_i=0}}^{r-1}\int_{\sum\limits_{i=0}^nz_i=0}
\mathcal{I}_V(\vecz, \vecy)
\mathcal{I}_q(\vecz, \vecy; \bvect, \bveca; \bvecs, \bvecb; B, n_B) .\label{I_SUE}
\end{gather}
The variables $(\vecz, \vecy)$ are integrated/summed over since they are associated with the gauge symmetry.
Inside the integrand, $\mathcal{I}_V$ and $\mathcal{I}_C$ are the one-loop contributions from the f\/ields
$V$ and~$q$,~$\bar{q}$, respectively, and are given by
\begin{gather*}
 \mathcal{I}_{V}(\vecz, \vecy)=
\frac{\lambda^{N_c-1}}{N_c!} \frac{1}{\prod\limits_{0\leq i<j\leq N_c-1}\Gamma(z_i-z_j,y_i-y_j)\,\Gamma(z_j-z_i,y_j-y_i)}, \\
 \mathcal{I}_{q}(\vecz, \vecy; \bvect, \bveca; \bvecs, \bvecb; B, n_B)  \\
 \qquad{} =\prod_{i=0}^{N_c-1} \prod_{j=0}^{N_f-1}
\overbrace{\Gamma(\bar{t}_j+z_i+B+r_q R,\bar{a}_j+y_i+n_B )}^{q}\,\overbrace{\Gamma(\bar{s}_j-z_i-B+r_q R,\bar{b}_j-y_i-n_B)}^{\bar{q}} .
\end{gather*}
As indicated above, the gamma function factors inside the expression for $\mathcal{I}_{q}$ are contributions from the quarks $q$, $\bar{q}$, and their symmetry charges are ref\/lected in the arguments of gamma functions. For example, the ${\rm U}(1)_B$ and ${\rm U}(1)_R$ charges of the quark $q$ are $+1$ and $r_q$, and hence the combination $B+ r_q R$ appears inside the gamma functions for the quark $q$.

Let us def\/ine the unbarred vectors $\vect$, $\vecs$, $\veca$, $\vecb$ by adding trace parts to the barred vectors $\bvect$, $\bvecs$, $\bveca$, $\bvecb$:
\begin{gather}
\vect =\bvect+ r_q R+B , \qquad\vecs=\bvecs+r_qR -B ,\qquad\veca=\bveca +n_B ,\qquad\vecb=\bvecb -n_B .\label{tsab}
\end{gather}
After this rewriting $B$ and $n_B$ are now included in the trace parts of $\vect$, $\vecs$, $\veca$, $\vecb$.
The one-loop determinant for the quark f\/ields $q$, $\bar{q}$ is written as
\begin{gather*}
\mathcal{I}_{q}(\vecz, \vecy; \vect, \veca; \vecs, \vecb)=\prod_{i=0}^{N_c-1} \prod_{j=0}^{N_f-1} \Gamma(t_j+z_i,a_j+y_i)\,\Gamma(s_j-z_i, b_j-y_i) ,
\end{gather*}
and the constraints \eqref{sum_2_} are written as
\begin{gather}
 \sum_{i=0}^{N_c-1} z_i=0 ,
\qquad \sum_{i=0}^{N_f-1} (t_i+s_i)=2 N_f r_q R=(N_f-N_c) (\tau+\sigma) , \nonumber\\
 \sum_{i=0}^{N_c-1} y_i \equiv \sum_{i=0}^{N_f-1} (a_i+b_i) \equiv 0 \quad(\textrm{mod } r) .\label{sum_2}
\end{gather}
We can now easily verify that the lens index \eqref{I_SUE} coincides with the sum/integral \eqref{AnIntDef} with $Z=Y=0$, with the identif\/ication $n=N_c-1$ and $m=N_f-N_c-1$:
\begin{gather}
(\mathcal{I}_e)^{N_f}_{{\rm SU}(N_c)}(\vect, \veca; \vecs, \vecb) =I^{N_f-N_c-1 }_{A_{N_c-1}}(0,0\,|\, \vect, \veca; \vecs, \vecb) .
\label{IE_A}
\end{gather}
Notice that the constraint \eqref{sum_2} also matches with \eqref{balancing} under this parameter identif\/ication.

{\bf Magnetic theory.} Let us next come to the magnetic theory. This is ${\rm SU}(\tilde{N}_c)$ SQCD with~$N_f$ f\/lavors, where the $N_c$ of the electric theory is replaced by
\begin{gather*}
\tilde{N}_c=N_f-N_c .
\end{gather*}
We have an ${\rm SU}(N_f-N_f)$ $\mathcal{N}=1$ vector multiplet $\tilde{V}$ and $N_f$ chiral multiplets $Q$ and $\bar{Q}$. In addition we also have a meson f\/ield $M$, and the non-trivial superpotential involving it:
\begin{gather*}
W_{\rm magnetic}= Q M \bar{Q} .
\end{gather*}

The f\/lavor symmetry of the magnetic theory is the same as that for the electric theory. The charge assignments of the f\/ields under gauge/f\/lavor symmetries is summarized in Table~\ref{SUM}. Note in particular that the ${\rm U}(1)_R$-charge $r_Q$ of the dual quarks $Q$, $\bar{Q}$ is the same as that for the quarks of the electric theory, under the substitution $N_c\to \tilde{N}_c$:
\begin{gather}
r_Q=r_q |_{N_c\to \tilde{N}_c} . \label{rqQ}
\end{gather}

\begin{table}[htbp]\centering
\caption{Magnetic theory for ${\rm SU}(N_c)$ Seiberg duality.}\label{SUM}
\vspace{2mm}
\begin{tabular}{c||c|c|c|c|c}
 & ${\rm SU}(\tilde{N}_c)$ & ${\rm SU}(N_f)_L$ &${\rm SU}(N_f)_R$& ${\rm U}(1)_B$ & ${\rm U}(1)_R$ \\
 \hline
 \hline
$Q$ & $\square$ & $\overline{\square}$ & $\bm{1}$ &$r_B=N_c/\tilde{N}_c$ &$r_Q=N_c/N_f$ \\
 \hline
$\bar{Q}$ & $\overline{\square}$ & $\bm{1}$ & $\square$ &$-r_B=-N_c/\tilde{N}_c$ & $r_Q=N_c/N_f$ \\
 \hline
 $M$ & $\bm{1}$ & $\square$ & $\overline{\square}$ &$0$ & $2r_q=2\left(1-N_c/N_f\right)$ \\
 \hline
 $\tilde{V}$ & adj.& $\bm{1}$ & $\bm{1}$&$0$ & $1$
\end{tabular}
\end{table}

Let us now come to the lens index of the theory. This is similar to that of the electric theory, but there are some important dif\/ferences. First we need to change $N_c\to N_f-N_c$. This is also ref\/lected in the change of the $R$-charge $r_q\to r_Q$, as stated in \eqref{rqQ}. Second, compared with the electric case we need to invert the signs of the ${\rm SU}(N_f)_L\times {\rm SU}(N_f)_R$ fugacities $(\bvect, \bvecs, \bveca, \bvecb)$; for example the electric quark $q$ was in the fundamental representation under the ${\rm SU}(N_f)_L$ symmetry, whereas the magnetic quark $Q$ is in the anti-fundamental representation under the same symmetry. We also need to take into account the dif\/ference in the ${\rm U}(1)_B$-symmetry charges, which has the ef\/fect of changing $B\to r_B B$. Finally, we also need to take into account contributions from mesons, which do not exist in the electric theory:
\begin{gather*}
\mathcal{I}_M(\bvect, \bveca; \bvecs, \bvecb) = \prod_{i,j=1}^{N_f} \Gamma(\bar{t}_i+\bar{s}_j +2 r_q R, \bar{a}_i+\bar{b}_j) .
\end{gather*}
By combining all the ingredients, we have
\begin{gather}
(\mathcal{I}_m)^{N_f}_{{\rm SU}(N_c)}(\bvect, \bveca; \bvecs, \bvecb)
=\mathcal{I}_M (\bvect, \bveca; \bvecs, \bvecb)
(\mathcal{I}_e)^{N_f}_{{\rm SU}(\tilde{N}_c)}(-\bvect, -\bveca; -\bvecs, -\bvecb; r_B B, r_B n_B) .
\label{IM_tmp}
\end{gather}

There is one subtlety here. The formula \eqref{IM_tmp} does not make sense as it is, since $r_B$ and hence $r_B n_B$ is in general not an integer.
This is related to the fact that charges under the ${\rm U}(1)$ symmetry (such as the ${\rm U}(1)_B$ symmetry as discussed here)
is not quantized~-- this means that the discrete holonomy for the torsion cycle of $S^1\times S^3/\mathbb{Z}_r$ can be turned on only if the charges of all the f\/ields under the ${\rm U}(1)$ symmetry is an integer. For our purposes, we can simply choose to take the integer parameter $n_B$ to be
\begin{gather}
n_B \in \tilde{N}_c \mathbb{Z} .\label{quantized}
\end{gather}
Equivalently, we choose the normalization of ${\rm U}(1)_B$ such that the charges of all the f\/ields under this symmetry are integers.

We can now include the trace part into the def\/initions of the vectors, as in the case of the electric theory:
\begin{gather}
\tvect =-\bvect+ r_QR +r_B B , \qquad
\tvecs=-\bvecs+r_QR-r_B B ,\nonumber\\
\tveca=-\bveca +r_B n_B ,\qquad
\tvecb=-\bvecb -r_B n_B ,\label{tsab_2}
\end{gather}
and these variables obey the constraints
\begin{gather}
 \sum_i z_i=0 ,\qquad \sum_i (\tilde{t}_i+\tilde{s}_i)=N_f r_Q (\tau+\sigma)=N_c (\tau+\sigma) , \nonumber\\
 \sum_i y_i \equiv \sum_i (a_i+b_i) \equiv 0 \quad(\textrm{mod } r) . \label{sum_2_}
\end{gather}
The magnetic lens index \eqref{IM_tmp} now reads
\begin{gather}
(\mathcal{I}_m)^{N_f}_{{\rm SU}(N_c)}
 =\left( \prod_{i,j=1}^{N_f}
\Gamma(t_i+s_j, a_i+b_j) \right)
(\mathcal{I}_e)^{N_f}_{{\rm SU}(\tilde{N}_c)}(\tvect, \tveca; \tvecs, \tvecb) \nonumber \\
\hphantom{(\mathcal{I}_m)^{N_f}_{{\rm SU}(N_c)}}{} = \left( \prod_{i,j=1}^{N_f}
\Gamma(t_i+s_j, a_i+b_j) \right) I_{A_{N_f-N_c-1}}^{N_c-1} (0,0\,|\,\tvect, \tveca; \tvecs, \tvecb) .
\label{IM_A}
\end{gather}

By comparing \eqref{tsab} and \eqref{tsab_2}, we can directly express the variables $(\tvect, \tvecs, \tveca, \tvecb)$ in terms of variables $(\vect, \vecs, \veca, \vecb)$ without tildes:
\begin{gather}
\tvect=\frac{T}{N_f-N_c} -\vect , \qquad
\tvecs=\frac{S}{N_f-N_c}-\vecs,\nonumber\\
\tveca=\frac{A}{N_f-N_c}-\veca,\qquad\tvecb=\frac{B}{N_f-N_c} -\vecb,\label{t_trans_SU}
\end{gather}
with
\begin{gather*}
 T=\sum_{i=1}^{N_f} t_i= N_f B+ (N_f-N_c) R , \qquad S=\sum_{i=1}^{N_f} t_i= -N_f B+ (N_f-N_c) R , \\
 A=\sum_{i=1}^{N_f} a_i=N_f n_B, \qquad B=\sum_{i=1}^{N_f} b_i=- N_f n_B .
\end{gather*}
Notice that $A/(N_f-N_c)$ and $B/(N_f-N_c)$ are integers thanks to the quantization condition~\eqref{quantized}.

{\bf Duality.} We are now ready to state the equality of the lens indices of the electric and magnetic theory, which simply states
\begin{gather}
(\mathcal{I}_e)^{N_f}_{{\rm SU}(N_c)}= (\mathcal{I}_m)^{N_f}_{{\rm SU}(N_c)} .
\label{IEM_A}
\end{gather}
In view of \eqref{IE_A} and \eqref{IM_A} we have by now shown
that this equation coincides with the transformation formula \eqref{transdef},
with $Z=Y=0$. Indeed, the relations between parameters as stated in \eqref{t_trans_SU} can be seen to coincide with the transformation rules \eqref{transrule}, after a straightforward change of variables, as follows from the quantization condition \eqref{quantized}.
This concludes our proof of the relation \eqref{IEM_A}.

The identity \eqref{IEM_A} has previously been checked only up to certain orders in a series-expansion
with respect to fugacities. The results of this section settles the problem of proving the identity mathematically in general, and
provides the one of the most
elaborate quantitative check of the Seiberg duality known to date.

\subsection[${\rm Sp}(2N_c)$]{$\boldsymbol{{\rm Sp}(2N_c)}$}

Let us next discuss Seiberg duality for ${\rm Sp}(2N_c)$ gauge groups.

{\bf Electric theory.} The electric theory is ${\rm Sp}(2N_c)$ theory\footnote{The convention for the ${\rm Sp}(2N_c)$ gauge group here is that its argument is always even, for example ${\rm Sp}(2)={\rm SU}(2)$.} with $N_f$ f\/lavors, which in practice means that we have a matter f\/ield $q$
in the fundamental $2N_c$-dimensional representation of~${\rm Sp}(2N_c)$. The superpotential is absent, as in the ${\rm SU}(N_c)$ case. This theory has ${\rm SU}(2N_f)\times {\rm U}(1)_R$ global symmetry, and the charge assignment of the f\/ields are listed in Table~\ref{SpE}.

\begin{table}[htbp]\centering
\caption{Electric theory for ${\rm Sp}(2N_c)$ Seiberg duality.}\label{SpE}
\vspace{2mm}

\begin{tabular}{c||c|c|c}
 & ${\rm Sp}(2N_c)$ & ${\rm SU}(2 N_f)$ & ${\rm U}(1)_R$ \\
 \hline
 \hline
 $q$ & $\square$ & $\square$ & $r'_q=1-(N+1)/N_f$ \\
 \hline
 $V$ & adj.& $1$ & $1$
\end{tabular}
\end{table}

As in the case of ${\rm SU}(N_c)$ theory, we need fugacities for the lens index, which are $(\vecz, \vecy)$ for ${\rm SU}(N_c)$ gauge symmetry and $(\bvect, \veca)$ for ${\rm SU}(2N_f)$ f\/lavor symmetry, with the constraints
\begin{gather}
\sum_{i=1}^{N_c-1} z_i=\sum_{i=1}^{N_f-1} \bar{t}_i=0 , \qquad
\sum_{i=1}^{N_c-1} y_i\equiv \sum_{i=1}^{N_f-1} a_i \equiv 0\quad (\textrm{mod } r) .
\label{sum_Sp}
\end{gather}
The lens index of this electric theory is then written as
\begin{gather}
(\mathcal{I}_e)^{N_f}_{{\rm Sp}(2 N_c)}(\bvect, \veca)
=\sum_{\substack{y_0,\ldots,y_{n-1}=0 \\ \sum\limits_{i=0}^{n}y_i=0}}^{r-1}\int_{\sum\limits_{i=0}^nz_i=0}
\mathcal{I}_V(\vecz, \vecy)
\mathcal{I}_q(\vecz, \vecy; \bvect, \veca) ,
\label{I_SpE}
\end{gather}
where $\mathcal{I}_V$ and $\mathcal{I}_C$ are the one-loop contributions from the f\/ields $V$ and $q$, respectively, and are given by
\begin{gather*}
 \mathcal{I}_{V}(\vecz, \vecy)=
\frac{\lambda^{N_c}}{2^{N_c} N_c!}
\frac{1}{\prod\limits_{i=0}^{N_c-1}\Gamma(\pm 2 z_i, \pm 2 y_i )}
\frac{1}{\prod\limits_{0\leq i<j\leq N_c-1}\Gamma(\pm z_i \pm z_j, \pm y_i \pm y_j)} , \\
 \mathcal{I}_{q}(\vecz, \vecy; \bvect, \veca) =\prod_{i=0}^{N_c-1} \prod_{j=0}^{2N_f-1}
\Gamma(\bar{t}_j\pm z_i+r'_q R, a_j\pm y_i ) , 
\end{gather*}
where the factor $2^{N_c} N_c!$ comes from the Weyl group for $C_{N_c}={\rm Sp}(2N_c)$.

We can def\/ine the unbarred vector $\bvect$ by
\begin{gather*}
\vect =\bvect+ r_q R , \qquad \veca=\bveca , 
\end{gather*}
and the constraints \eqref{sum_Sp} are written as
\begin{gather*}
\sum_{i=0}^{N_c-1} z_i=0 ,
\qquad \sum_{i=0}^{2N_f-1} t_i=2 N_f r'_q R=(N_f-N_c-1) (\tau+\sigma) , \\
\sum_{i=0}^{N_c-1} y_i\equiv \sum_{i=0}^{N_f-1} a_i \equiv 0 \quad(\textrm{mod } r) .
\end{gather*}

We can now easily verify that the lens index \eqref{I_SpE} coincides with the sum/integral \eqref{BCnIntDef}, with the identif\/ication $n=N_c$ and $m=N_f-N_c-2$:
\begin{gather}
(\mathcal{I}_e)^{N_f}_{{\rm Sp}(2N_c)}(\vect, \veca) =I^{N_f-N_c-2 }_{BC_{N_c}}(\vect, \veca) .\label{IE_BC}
\end{gather}
Notice that the constraint \eqref{sum_2} also matches with \eqref{balancing2}
under this parameter identif\/ication.

{\bf Magnetic theory.} Let us next discuss the magnetic theory. We will be brief here since the analysis is similar to previous cases.

The f\/ields, symmetries and charge assignments are summarized in Table~\ref{SpM}. Similar to the case of the ${\rm SU}(N_c)$ theory we have a meson f\/ield $M$ with the superpotential $W=M Q Q$, except now $M$ is in the anti-symmetric representation under the f\/lavor ${\rm SU}(2N_f)$ symmetry.
Another important dif\/ference from the ${\rm SU}(N_c)$ case is that the value of $N_c$ in the magnetic theory is given by
\begin{gather*}
\tilde{N}_c=N_f-N_c-2 .
\end{gather*}

\begin{table}[htbp]\centering
\caption{Magnetic theory for ${\rm Sp}(2N_c)$ Seiberg duality.}\label{SpM}
\vspace{2mm}

\begin{tabular}{c||c|c|c}
 & ${\rm Sp}(2\tilde{N}_c)$ & ${\rm SU}(2 N_f)$ & ${\rm U}(1)_R$ \\
 \hline
 \hline
 $Q$ & $\square$ & $\overline{\square}$ & $r'_Q=(N_c+1)/N_f$ \\
 \hline
 $\tilde{V}$ & adj.& $1$ & $1$\\
 \hline
 $M$ & $1$& $\textrm{anti-symm.}$ & $2(\tilde{N}_c+1)/N_f$
\end{tabular}
\end{table}

By repeating the similar manipulations as in the previous cases, we obtain the lens index of the magnetic theory to be
\begin{gather}
(\mathcal{I}_m)^{N_f}_{{\rm Sp}(2N_c)}(\tvect, \veca) =
\left( \prod_{i,j=1}^{N_f}
\Gamma(t_i+t_j, a_i+a_j) \right)
I^{N_c}_{BC_{N_f-N_c-2}}(\vect, -\veca) , \label{IM_BC}
\end{gather}
where we def\/ined
\begin{gather*}
\tvect = -\bvect +r'_Q R = -\vect +(r'_q+r'_Q) R=-\vect +\frac{\tau+\sigma}{2} .
\end{gather*}

{\bf Duality.} We can now easily check from \eqref{IE_BC} and \eqref{IM_BC} that the duality relation
\begin{gather*}
(\mathcal{I}_e)^{N_f}_{{\rm Sp}(2N_c)}=(\mathcal{I}_m)^{N_f}_{{\rm Sp}(2N_c)},
\end{gather*}
reduces to the $BC_n$ elliptic hypergeometric sum/integral transformation formula as stated in~\eqref{transdef2}. This is what we wanted to show.

It would be interesting to prove identities for more general Seiberg dualities, for more general gauge groups and more general matters (for example matters in spinor representations). The case of Seiberg dualities for ${\rm SO}(N)$ and ${\rm Spin}(N)$ gauge groups is currently under investigation.

\section{Application to integrable lattice models}\label{sec:YBE}

Let us now come to our second application for the $A_n$ elliptic hypergeometric sum/integral transformation formula \eqref{transdef}, this time to integrable lattice models.

A particular case of Theorem \ref{mainthm} (when $m=n$) is equivalent to an identity in statistical mechanics known as the star-star relation. The star-star relation is a particular condition of integrability for lattice models of statistical mechanics, which implies that the Boltzmann weights of the model satisfy the Yang--Baxter equation. This in turn implies that the row-to-row transfer matrices of the lattice model commute in pairs \cite{Baxter:1997tn}, allowing for an exact solution of the model.

The $r=1$ case of Theorem \ref{mainthm} \cite{RainsT} was previously shown \cite{Bazhanov:2013bh} to imply a multi-spin solution of the star-star relation obtained by Bazhanov and Sergeev \cite{Bazhanov:2011mz}. The general $r\geq1$ solution of the star-star relation corresponding to Theorem \ref{mainthm} was discovered by the second author \cite{Yamazaki:2013nra} using the gauge/YBE correspondence \cite{Terashima:2012cx,Yamazaki:2012cp,Yamazaki:2013nra} between 2d integrable lattice models and 4d $\mathcal{N}=1$ supersymmetric gauge theories. In this section this lattice model and corresponding star-star relation are introduced, and it is explicitly shown that the latter star-star relation reduces to the $m=n$ case of Theorem \ref{mainthm}.

\subsection{Square lattice model}

Let us f\/irst def\/ine the lattice model of statistical mechanics. Denote the square lattice by~$L$, consisting of a set of vertices and a set of edges, the latter denoted respectively by $V(L)$ and~$E(L)$. An edge $(ij)\in E(L)$ connects two vertices $i,j\in V(L)$.

The square lattice $L$ is shown in Fig.~\ref{fig-lattice}, along with a directed rapidity lattice, which will be denoted by $\mathscr{L}$. The rapidity lattice $\mathscr{L}$ is made up of four dif\/ferent types of directed rapidity lines, which are distinguished by their orientation (horizontal or vertical), and by whether they are solid or dashed lines. The four types of rapidity lines are labelled by four real valued rapidity variables $u$, $u'$, $v$, $v'$.

\begin{figure}[htb]
\centering
\begin{tikzpicture}[scale=1]
\draw[-] (-0.5,-0.5)--(3.5,3.5);
\draw[-] (-0.5,3.5)--(3.5,-0.5);
\draw[-] (-0.5,1.5)--(1.5,3.5)--(3.5,1.5)--(1.5,-0.5)--(-0.5,1.5);
\draw[-] (-4.5,-0.5)--(-0.5,3.5);
\draw[-] (-4.5,3.5)--(-0.5,-0.5);
\draw[-] (-4.5,1.5)--(-2.5,3.5)--(-0.5,1.5)--(-2.5,-0.5)--(-4.5,1.5);
\foreach \x in {-4,-2,...,2}{
\draw[->,very thin] (\x,-1) -- (\x,4);
\fill[white!] (\x,-1) circle (0.08pt)
node[below=3.1pt]{\color{black}\small $v$};}
\foreach \x in {-3,-1,...,3}{
\draw[->,thick,dashed] (\x,-1) -- (\x,4);
\fill[white!] (\x,-1) circle (0.08pt)
node[below=0.05pt]{\color{black}\small $v'$};}
\foreach \y in {1,3}{
\draw[->,thick,dashed] (-5,\y) -- (4,\y);
\fill[white!] (-5,\y) circle (0.08pt)
node[left=0.05pt]{\color{black}\small $u'$};}
\foreach \y in {0,2}{
\draw[->,very thin] (-5,\y) -- (4,\y);
\fill[white!] (-5,\y) circle (0.08pt)
node[left=2.9pt]{\color{black}\small $u$};}
\foreach \y in {-0.5,1.5,3.5}{
\foreach \x in {-4.5,-2.5,...,3.5}{
\filldraw[fill=black,draw=black] (\x,\y) circle (3.0pt);}}
\foreach \y in {0.5,2.5}{
\foreach \x in {-3.5,-1.5,...,2.5}{
\filldraw[fill=black,draw=black] (\x,\y) circle (3.0pt);}}
\end{tikzpicture}

\caption{The square lattice $L$ drawn diagonally, and its medial rapidity lattice $\mathscr{L}$, consisting of directed lines labelled by the four rapidity variables $u$, $u'$, $v$, $v'$.}
\label{fig-lattice}
\end{figure}
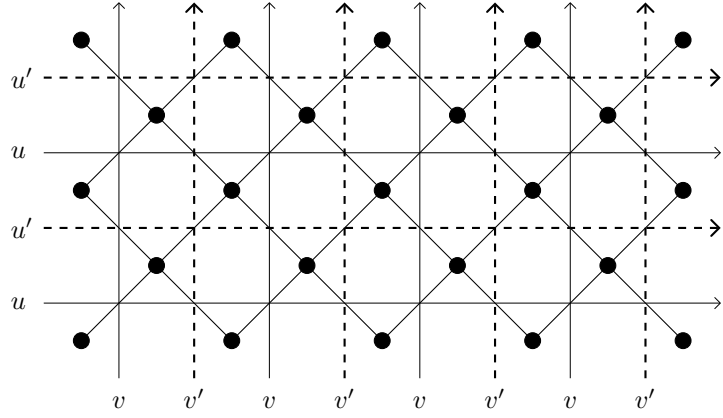

Spin variables $\sigma_i$ are assigned to vertices $i\in V(L)$, and take values
\begin{gather*}
\sigma_i=\{(x_{i,1},m_{i,1}),(x_{i,2},m_{i,2}),\ldots,(x_{i,n},m_{i,n})\},
\end{gather*}
where $n=1,2,\ldots$, and
\begin{gather*}
x_{i,j}\in[0,2\pi),\qquad m_{i,j}\in\mathbb{Z}_r,\qquad j=1,2,\ldots,n,
\end{gather*}
are respectively the real and discrete valued components of a spin $\sigma_i$. These spin components are subject to the constraints
\begin{gather}\label{measure}
\sum_{a=1}^nx_{i,a}=0,\qquad\sum_{a=1}^nm_{i,a}=0.
\end{gather}
Taking these constraints into account, the integration measure for a spin $\sigma_i$ is denoted by
\begin{gather}\label{measure2}
\int {\rm d}\sigma_i:= \hspace{-7mm} \sum_{\substack{m_{i,1}=0\\ \hspace{1.3cm}\sum\limits_{j=1}^{n}m_{i,j}=0}}^{r-1}\hspace{-1.0cm}\cdots \sum_{m_{i,n-1}=0}^{r-1}\;\;\int_{\substack{~\\[0.45cm] \hspace{-1.25cm}0\\ \sum\limits_{j=1}^{n}x_{i,j}=0}}^1 \hspace{-1.0cm}\cdots \int^1_0\quad \prod_{k=1}^{n-1}{{\rm d}x_{i,k}}.
\end{gather}

The crossing of rapidity lines on edges $(ij)\in E(L)$ in Fig.~\ref{fig-lattice} distinguish the four dif\/ferent types of edges shown explicitly in Fig.~\ref{fig-crosses}. The two types of Boltzmann weights $W_\alpha(\sigma_i,\sigma_j)$, $\olW_\alpha(\sigma_i,\sigma_j)$ are assigned to the four types of edges as indicated in Fig.~\ref{fig-crosses}, and depend on the value of the spins at the vertices, and the value of the rapidity variables crossing the edge. The ordering of spins variables matters here, i.e., in general
\begin{gather*}
W_{\alpha}(\sigma_i,\sigma_j)\neq W_{\alpha}(\sigma_j,\sigma_i),\qquad\olW_\alpha(\sigma_i,\sigma_j)\neq \olW_\alpha(\sigma_j,\sigma_i).
\end{gather*}

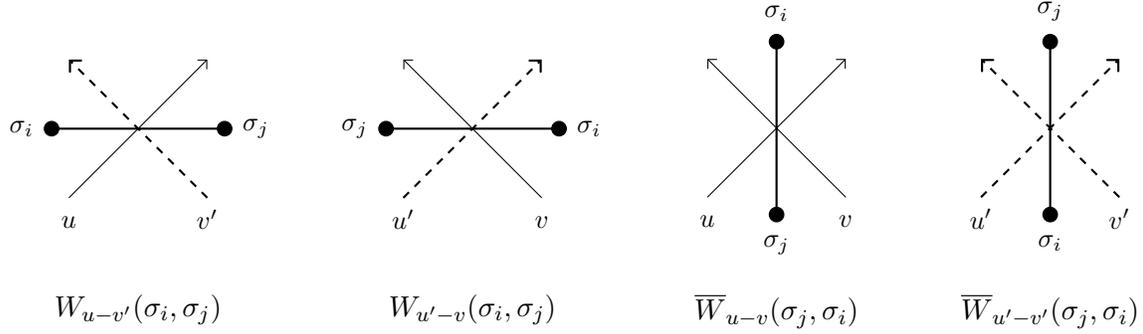
\begin{figure}[htb]
\centering
\begin{tikzpicture}[scale=2.3]
\draw[-,thick] (-0.5,2)--(0.5,2);
\draw[->,thick,dashed] (0.4,1.6)--(-0.4,2.4);
\fill[white!] (0.4,1.6) circle (0.01pt)
node[below=0.5pt]{\color{black}\small $v'$};
\draw[->] (-0.4,1.6)--(0.4,2.4);
\fill[white!] (-0.4,1.6) circle (0.01pt)
node[below=3.1pt]{\color{black}\small $u$};
\filldraw[fill=black,draw=black] (-0.5,2) circle (1.2pt)
node[left=3pt]{\color{black}\small $\sigma_i$};
\filldraw[fill=black,draw=black] (0.5,2) circle (1.2pt)
node[right=3pt]{\color{black}\small $\sigma_j$};

\fill (0,1.1) circle(0.01pt)
node[below=0.05pt]{\color{black} $W_{u-v'}(\sigma_i,\sigma_j)$};

\begin{scope}[xshift=55pt]
\draw[-,thick] (-0.5,2)--(0.5,2);
\draw[->] (0.4,1.6)--(-0.4,2.4);
\fill[white!] (0.4,1.6) circle (0.01pt)
node[below=3.3pt]{\color{black}\small $v$};
\draw[->,thick,dashed] (-0.4,1.6)--(0.4,2.4);
\fill[white!] (-0.4,1.6) circle (0.01pt)
node[below=0.5pt]{\color{black}\small $u'$};
\filldraw[fill=black,draw=black] (-0.5,2) circle (1.2pt)
node[left=3pt]{\color{black}\small $\sigma_j$};
\filldraw[fill=black,draw=black] (0.5,2) circle (1.2pt)
node[right=3pt]{\color{black}\small $\sigma_i$};

\fill (0,1.1) circle(0.01pt)
node[below=0.05pt]{\color{black} $W_{u'-v}(\sigma_i,\sigma_j)$};
\end{scope}

\begin{scope}[xshift=105pt,yshift=57pt]
\draw[-,thick] (0,-0.5)--(0,0.5);
\draw[->] (-0.4,-0.4)--(0.4,0.4);
\fill[white!] (-0.4,-0.4) circle (0.01pt)
node[below=3.2pt]{\color{black}\small $u$};
\draw[->] (0.4,-0.4)--(-0.4,0.4);
\fill[white!] (0.4,-0.4) circle (0.01pt)
node[below=3.2pt]{\color{black}\small $v$};
\filldraw[fill=black,draw=black] (0,-0.5) circle (1.2pt)
node[below=4pt]{\color{black}\small $\sigma_j$};
\filldraw[fill=black,draw=black] (0,0.5) circle (1.2pt)
node[above=4pt]{\color{black}\small $\sigma_i$};

\fill (0,-0.9) circle(0.01pt)
node[below=0.05pt]{\color{black} $\olW_{u-v}(\sigma_j,\sigma_i)$};
\end{scope}

\begin{scope}[xshift=150pt,yshift=57pt]
\draw[-,thick] (0,-0.5)--(0,0.5);
\draw[->,thick,dashed] (-0.4,-0.4)--(0.4,0.4);
\fill[white!] (-0.4,-0.4) circle (0.01pt)
node[below=0.5pt]{\color{black}\small $u'$};
\draw[->,thick,dashed] (0.4,-0.4)--(-0.4,0.4);
\fill[white!] (0.4,-0.4) circle (0.01pt)
node[below=0.5pt]{\color{black}\small $v'$};
\filldraw[fill=black,draw=black] (0,-0.5) circle (1.2pt)
node[below=4pt]{\color{black}\small $\sigma_i$};
\filldraw[fill=black,draw=black] (0,0.5) circle (1.2pt)
node[above=4pt]{\color{black}\small $\sigma_j$};

\fill (0,-0.9) circle(0.01pt)
node[below=0.05pt]{\color{black} $\olW_{u'-v'}(\sigma_j,\sigma_i)$};
\end{scope}

\end{tikzpicture}
\caption{Four dif\/ferent types of edges and their associated Boltzmann weights in \eqref{BWdef}.}\label{fig-crosses}
\end{figure}

The Boltzmann weights are conveniently expressed in terms of a function $\Phi(z,m)$, which is def\/ined in terms of the lens elliptic gamma function \eqref{legf2} as
\begin{gather*}
\Phi(z,m)=\Gamma\left(\frac{\sigma+\tau}{2}-z,-m;\sigma,\tau\right).
\end{gather*}
Then the Boltzmann weights are given by
\begin{gather}
 W_{\alpha}(\sigma_a,\sigma_b) = \prod_{i,j=1}^n\Phi(x_{a,i}-x_{b,j}+\ii\alpha,m_{a,i}-m_{b,j}), \nonumber\\
 \overline{W}_\alpha(\sigma_a,\sigma_b) = \sqrt{S(\sigma_a)S(\sigma_b)}\,W_{\eta-\alpha}(\sigma_a,\sigma_b),\label{BWdef}
\end{gather}
where
\begin{gather*}
S(\sigma_a)=\prod_{1\leq i<j\leq n}\Phi(-\ii\eta+x_{a,i}-x_{a,j},m_{a,i}-m_{a,j})\,\Phi(-\ii\eta+x_{a,j}-x_{a,i},m_{a,j}-m_{a,i}).
\end{gather*}
The parameter $\eta$ is known as the crossing parameter, and it relates the Boltzmann weight $W_\alpha(\sigma_a,\sigma_b)$ to the Boltzmann weight $\overline{W}_\alpha(\sigma_a,\sigma_b)$. It is given here by
\begin{gather*}
\eta=-\frac{\ii}{2}(\sigma+\tau).
\end{gather*}

Note that the Boltzmann weights \eqref{BWdef} physically represent an interaction energy between the two spins connected by the edge of the lattice. It is therefore often desirable that Boltzmann weights are positive and real valued, however the conditions for a regime where this condition is satisf\/ied are unfortunately not known.

Now let $E^{(1)}$, $E^{(2)}$, $E^{(3)}$, $E^{(4)}$, be respectively the sets of the four types of edges of $L$, that are depicted from left to right in Fig.~\ref{fig-crosses}. Then the partition function of the model is def\/ined as
\begin{gather}
Z=\int\prod_{(ij)\in E^{(1)}(L)}W_{u-v'}(\sigma_i,\sigma_j)\prod_{(ij)\in E^{(2)}(L)} W_{u'-v}(\sigma_i,\sigma_j)\prod_{(ij)\in E^{(3)}(L)} \olW_{u-v}(\sigma_i,\sigma_j)\nonumber\\
\hphantom{Z=}{}\times \prod_{(ij)\in E^{(4)}(L)} \olW_{u'-v'}(\sigma_i,\sigma_j) \prod_{i=1}^N{\rm d}\sigma_{i},\label{Zdefm}
\end{gather}
where the products are taken over the four types of edges of $L$ given in Fig.~\ref{fig-crosses}, and the integration is taken over $N$ spins $\sigma_i$ interior to the lattice, with boundary spins kept f\/ixed. Note that the integration is taken with respect to \eqref{measure2}.

In addition to the edge formulation given in Fig.~\ref{fig-crosses}, the model may be formulated as an interaction-round-a-face (IRF) model \cite{Baxter:1982zz}, in terms of either of the four-edge stars depicted in Fig.~\ref{fig-IRF}. These stars are associated two dif\/ferent Boltzmann weights $W^{(1)}_{\bu\bv}$, $W^{(2)}_{\bu\bv}$, according to Fig.~\ref{fig-crosses},
and are given by the expressions
\begin{gather}
W^{(1)}_{\bu\bv}\left( \begin{matrix}\sigma_i&\sigma_j\\ \sigma_k&\sigma_l\end{matrix}\right)=
\int {\rm d}\sigma_h \overline{W}_{u-v}(\sigma_k,\sigma_h)\overline{W}_{u'-v'}(\sigma_j,\sigma_h)W_{u'-v}(\sigma_h,\sigma_i)W_{u-v'}(\sigma_h,\sigma_l),\nonumber\\
W^{(2)}_{\bu\bv}\left( \begin{matrix}\sigma_i&\sigma_j\\ \sigma_k&\sigma_l\end{matrix} \right)=
\int {\rm d}\sigma_h \overline{W}_{u-v}(\sigma_h,\sigma_j)\overline{W}_{u'-v'}(\sigma_h,\sigma_k)W_{u'-v}(\sigma_l,\sigma_h)W_{u-v'}(\sigma_i,\sigma_h),
\label{IRFweights}
\end{gather}
where $\bu=\{u,u'\}$, and $\bv=\{v,v'\}$.

\begin{figure}[htb]
\centering
\begin{tikzpicture}[scale=1.2]

\draw[-,thick] (-1,-1)--(1,1);
\draw[-,thick] (1,-1)--(-1,1);

\filldraw[fill=black!,draw=black!] (-1,-1) circle (2.2pt)
node[below=1.5pt]{\color{black}\small $\sigma_k$};
\filldraw[fill=black!,draw=black!] (1,-1) circle (2.2pt)
node[below=1.5pt]{\color{black}\small $\sigma_l$};
\filldraw[fill=black!,draw=black!] (-1,1) circle (2.2pt)
node[above=1.5pt]{\color{black}\small $\sigma_i$};
\filldraw[fill=black!,draw=black!] (1,1) circle (2.2pt)
node[above=1.5pt]{\color{black}\small $\sigma_j$};

\filldraw[fill=black!,draw=black!] (0,0) circle (2.2pt)
node[below=2.5pt]{\color{black}\small $\sigma_h$};

\draw[->,thick,dashed] (-1.2,0.5) -- (1.2,0.5);
\draw[white!] (-1.2,0.5) circle (0.01pt)
node[left=1.5pt]{\color{black}\small $u'$};
\draw[->] (-1.2,-0.5) -- (1.2,-0.5);
\draw[white!] (-1.3,-0.5) circle (0.01pt)
node[left=1.5pt]{\color{black}\small $u$};
\draw[->] (-0.5,-1.2) -- (-0.5,1.2);
\draw[white!] (-0.5,-1.27) circle (0.01pt)
node[below=1.5pt]{\color{black}\small $v$};
\draw[->,thick,dashed] (0.5,-1.2) -- (0.5,1.2);
\draw[white!] (0.5,-1.2) circle (0.01pt)
node[below=1.5pt]{\color{black}\small $v'$};

\draw[white!] (0,-1.8) circle (0.01pt)
node[below=0.1pt]{\color{black}\small $W^{(1)}_{\bu\bv}\left( \begin{matrix}\sigma_i&\sigma_j\\ \sigma_k&\sigma_l\end{matrix}\right)$};

\begin{scope}[xshift=150pt]

\draw[-,thick] (-1,-1)--(1,1);
\draw[-,thick] (1,-1)--(-1,1);

\filldraw[fill=black!,draw=black!] (-1,-1) circle (2.2pt)
node[below=1.5pt]{\color{black}\small $\sigma_k$};
\filldraw[fill=black!,draw=black!] (1,-1) circle (2.2pt)
node[below=1.5pt]{\color{black}\small $\sigma_l$};
\filldraw[fill=black!,draw=black!] (-1,1) circle (2.2pt)
node[above=1.5pt]{\color{black}\small $\sigma_i$};
\filldraw[fill=black!,draw=black!] (1,1) circle (2.2pt)
node[above=1.5pt]{\color{black}\small $\sigma_j$};

\filldraw[fill=black!,draw=black!] (0,0) circle (2.2pt)
node[below=2.5pt]{\color{black}\small $\sigma_h$};

\draw[->] (-1.2,0.5) -- (1.2,0.5);
\draw[white!] (-1.3,0.5) circle (0.01pt)
node[left=1.5pt]{\color{black}\small $u$};
\draw[->,thick,dashed] (-1.2,-0.5) -- (1.2,-0.5);
\draw[white!] (-1.2,-0.5) circle (0.01pt)
node[left=1.5pt]{\color{black}\small $u'$};
\draw[->,thick,dashed] (-0.5,-1.2) -- (-0.5,1.2);
\draw[white!] (-0.5,-1.2) circle (0.01pt)
node[below=1.5pt]{\color{black}\small $v'$};
\draw[->] (0.5,-1.2) -- (0.5,1.2);
\draw[white!] (0.5,-1.27) circle (0.01pt)
node[below=1.5pt]{\color{black}\small $v$};

\draw[white!] (0,-1.8) circle (0.01pt)
node[below=0.1pt]{\color{black}\small $W^{(2)}_{\bu\bv}\left( \begin{matrix}\sigma_i&\sigma_j\\ \sigma_k&\sigma_l\end{matrix}\right)$};

\end{scope}
\end{tikzpicture}

\caption{Two types of four-edge stars and associated Boltzmann weights in \eqref{IRFweights}.}
\label{fig-IRF}
\end{figure}
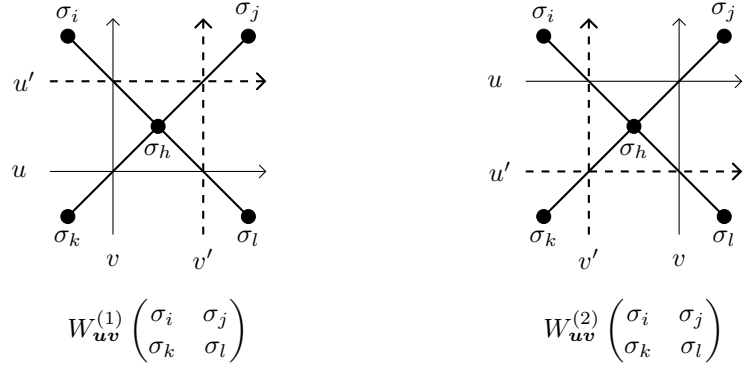

The lattice $L$ may be produced by periodic translations of either one of the two four-edge stars depicted in Fig.~\ref{fig-IRF}. The partition function \eqref{Zdefm} may then be written in equivalent forms (up to boundary ef\/fects), in terms of either of the
two Boltzmann weights in \eqref{IRFweights}. For example, let $V^{(1)}$, and $V^{(2)}$, denote two disjoint subsets of $V$, where $V^{(1)}$ is the set of all vertices of the type associated to $\sigma_h$ on the left hand side of Fig.~\ref{fig-IRF}, and $V^{(2)}$ is the set of vertices of the remaining type on the right hand side of Fig.~\ref{fig-IRF}. Then \eqref{Zdefm} may be written as
\begin{gather}\label{Z-IRF}
Z=\int \prod_{h\in V^{(1)}} W^{(1)}_{\bu\bv}\left( \begin{matrix}\sigma_i&\sigma_j\\ \sigma_k&\sigma_l\end{matrix}\right) \prod_{a\in V^{(2)}} {\rm d}\sigma_a.
\end{gather}
Note that the integration over vertices $h\in V^{(1)}$ is already made through the def\/inition of $W^{(1)}_{\bu\bv}$ in \eqref{IRFweights}, and the expression \eqref{Z-IRF} contains the integration over the remaining internal vertices $a\in V^{(2)}$.

\subsection[Star-star relation and $A_n$ sum/integral transformation]{Star-star relation and $\boldsymbol{A_n}$ sum/integral transformation}

An important property of the Boltzmann weights \eqref{IRFweights} is that they satisfy the following {\it star-star relation}
\begin{gather}
W_{v'-v}(\sigma_l,\sigma_k) W_{u'-u}(\sigma_l,\sigma_j) W^{(1)}_{\bu\bv} \left( \begin{matrix}\sigma_i&\sigma_j\\ \sigma_k&\sigma_l\end{matrix} \right)\nonumber\\
\qquad{} =W_{v'-v}(\sigma_j,\sigma_i) W_{u'-u}(\sigma_k,\sigma_i) W^{(2)}_{\bu\bv} \left( \begin{matrix}\sigma_i&\sigma_j\\ \sigma_k&\sigma_l\end{matrix} \right) .\label{ssr}
\end{gather}
This relation is depicted graphically in Fig.~\ref{ssfig}.

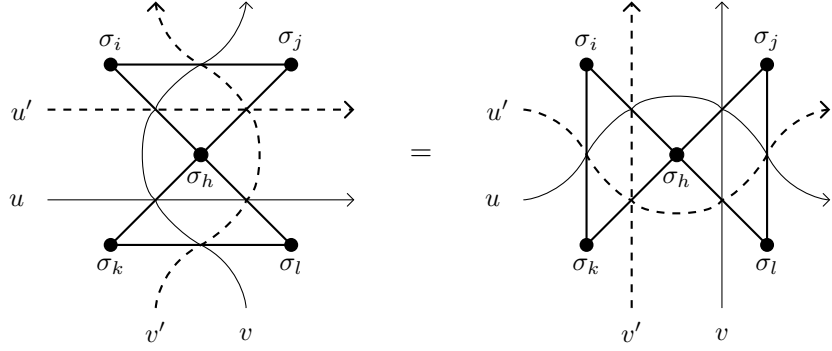
\begin{figure}[htb]\centering
\begin{tikzpicture}[scale=1.2]

\draw[-,thick] (-1,-1)--(1,1);
\draw[-,thick] (1,-1)--(-1,1);
\draw[-,thick] (-1,1)--(1,1);
\draw[-,thick] (-1,-1)--(1,-1);

\fill[black!] (-1,-1) circle (2.2pt)
node[below=1.5pt]{\color{black}\small $\sigma_k$};
\fill[black!] (1,-1) circle (2.2pt)
node[below=1.5pt]{\color{black}\small $\sigma_l$};
\fill[black!] (-1,1) circle (2.2pt)
node[above=1.5pt]{\color{black}\small $\sigma_i$};
\fill[black!] (1,1) circle (2.2pt)
node[above=1.5pt]{\color{black}\small $\sigma_j$};

\filldraw[fill=black!,draw=black!] (0,0) circle (2.2pt)
node[below=2.5pt]{\color{black}\small $\sigma_h$};

\draw[->,thick,dashed] (-1.7,0.5) -- (1.7,0.5);
\draw[white!] (-1.7,0.5) circle (0.01pt)
node[left=1.5pt]{\color{black}\small $u'$};
\draw[->] (-1.7,-0.5) -- (1.7,-0.5);
\draw[white!] (-1.8,-0.5) circle (0.01pt)
node[left=1.5pt]{\color{black}\small $u$};
\draw[->,thick,dashed] (-0.5,-1.7) .. controls (-0.5,-1.6) and (-0.4,-1.2) .. (0,-1) .. controls (0.2,-0.9) and (0.5,-0.6) .. (0.5,-0.5) .. controls (0.7,-0.4) and (0.7,0.4) .. (0.5,0.5) .. controls (0.5,0.6) and (0.2,0.9) .. (0,1) .. controls (-0.4,1.2) and (-0.5,1.6) .. (-0.5,1.7);
\draw[white!] (-0.5,-1.7) circle (0.01pt)
node[below=1.5pt]{\color{black}\small $v'$};
\draw[->] (0.5,-1.7) .. controls (0.5,-1.6) and (0.4,-1.2) .. (0,-1) .. controls (-0.2,-0.9) and (-0.5,-0.6) .. (-0.5,-0.5) .. controls (-0.7,-0.4) and (-0.7,0.4) .. (-0.5,0.5) .. controls (-0.5,0.6) and (-0.2,0.9) .. (0,1) .. controls (0.4,1.2) and (0.5,1.6) .. (0.5,1.7);
\draw[white!] (0.5,-1.77) circle (0.01pt)
node[below=1.5pt]{\color{black}\small $v$};

\draw[white!] (2.2,0) circle (0.01pt)
node[right=0.1pt]{\color{black}=};

\begin{scope}[xshift=150pt]

\draw[-,thick] (-1,-1)--(1,1);
\draw[-,thick] (1,-1)--(-1,1);
\draw[-,thick] (-1,-1)--(-1,1);
\draw[-,thick] (1,-1)--(1,1);

\fill[black!] (-1,-1) circle (2.2pt)
node[below=1.5pt]{\color{black}\small $\sigma_k$};
\fill[black!] (1,-1) circle (2.2pt)
node[below=1.5pt]{\color{black}\small $\sigma_l$};
\fill[black!] (-1,1) circle (2.2pt)
node[above=1.5pt]{\color{black}\small $\sigma_i$};
\fill[black!] (1,1) circle (2.2pt)
node[above=1.5pt]{\color{black}\small $\sigma_j$};

\filldraw[fill=black!,draw=black!] (0,0) circle (2.2pt)
node[below=2.5pt]{\color{black}\small $\sigma_h$};

\draw[->,thick,dashed] (-1.7,0.5) .. controls (-1.6,0.5) and (-1.2,0.4) .. (-1,0) .. controls (-0.9,-0.2) and (-0.6,-0.5) .. (-0.5,-0.5) .. controls (-0.4,-0.7) and (0.4,-0.7) .. (0.5,-0.5) .. controls (0.6,-0.5) and (0.9,-0.2) .. (1,0) .. controls (1.2,0.4) and (1.6,0.5) .. (1.7,0.5);
\draw[white!] (-1.7,0.5) circle (0.01pt)
node[left=1.5pt]{\color{black}\small $u'$};
\draw[->] (-1.7,-0.5) .. controls (-1.6,-0.5) and (-1.2,-0.4) .. (-1,0) .. controls (-0.9,0.2) and (-0.6,0.5) .. (-0.5,0.5) .. controls (-0.4,0.7) and (0.4,0.7) .. (0.5,0.5) .. controls (0.6,0.5) and (0.9,0.2) .. (1,0) .. controls (1.2,-0.4) and (1.6,-0.5) .. (1.7,-0.5);
\draw[white!] (-1.8,-0.5) circle (0.01pt)
node[left=1.5pt]{\color{black}\small $u$};
\draw[->,thick,dashed] (-0.5,-1.7) -- (-0.5,1.7);
\draw[white!] (-0.5,-1.7) circle (0.01pt)
node[below=1.5pt]{\color{black}\small $v'$};
\draw[->] (0.5,-1.7) -- (0.5,1.7);
\draw[white!] (0.5,-1.77) circle (0.01pt)
node[below=1.5pt]{\color{black}\small $v$};

\end{scope}
\end{tikzpicture}

\caption{The star-star relation \eqref{ssr}. This provides a graphical representation of the $A_n$ transformation in Theorem \ref{mainthm}, for the case $m=n$.}
\label{ssfig}
\end{figure}

The particular solution of the star-star relation \eqref{ssr} given by Boltzmann weights~\eqref{BWdef} and~\eqref{IRFweights} was given by the second author \cite{Yamazaki:2013nra}.

The main result for this section is showing that the star-star relation \eqref{ssr} is equivalent to Theorem \ref{mainthm} in the case $m=n$. This ends up being rather straightforward.

Indeed, consider the new variables
\begin{alignat}{3}
& t_j=\ii(u-v)- x_{c,j},\qquad &&s_j=-(u'-v-\eta)+ x_{a,j}, & \nonumber\\
& t_{n+j}=\ii(u'-v')- x_{b,j}, \qquad && s_{n+j}=-(u-v'-\eta)+ x_{d,j},& \label{ssrvar1}
\end{alignat}
and
\begin{gather}\label{ssrvar2}
a_j=-m_{c,j},\qquad a_{n+j}=-m_{b,j},\qquad b_j=+m_{a,j},\qquad b_{n+j}=+m_{d,j},
\end{gather}
where the $x_i$, and the $m_i$ satisfy \eqref{measure}.

Then the Boltzmann weights \eqref{IRFweights}, are seen to be equivalent to the $A_n$ sum/integrals \eqref{AnIntDef} in the form
\begin{gather}
\label{ssrlhs}
W^{(1)}_{\bu\bv}\left(\begin{matrix}\sigma_a&\sigma_b\\ \sigma_c&\sigma_d\end{matrix}\right)=I^{n-1}_{A_{n-1}}(\vect,\veca;\vecs,\vecb),
\end{gather}
and
\begin{gather}\label{ssrrhs}
W^{(2)}_{\bu\bv}\left(\begin{matrix}\sigma_a&\sigma_b\\ \sigma_c&\sigma_d\end{matrix}\right)=I^{n-1}_{A_{n-1}}(\tvect,\tveca;\tvecs,\tvecb),
\end{gather}
where the $\tvect$, $\tveca$, $\tvecs$, $\tvecb$ are the variables \eqref{ssrvar1}, \eqref{ssrvar2}, transformed according to \eqref{transrule}.

We also have
\begin{gather}\label{ssrfac}
\frac{W_{v'-v}(\sigma_j,\sigma_i)\,W_{u'-u}(\sigma_k,\sigma_i)}{W_{v'-v}(\sigma_l,\sigma_k)\,W_{u'-u}(\sigma_l,\sigma_j)}
=\prod_{i,j=1}^n\Gamma(t_j+s_i,a_j+b_i).
\end{gather}

The star-star relation \eqref{ssr} then follows from Theorem \ref{mainthm} with \eqref{ssrlhs}, \eqref{ssrrhs}, and \eqref{ssrfac}.

\subsection*{Acknowledgements}

The main results in Theorem \ref{mainthm}, and Theorem \ref{secondthm}, were presented in March 2017 at the workshop ``Elliptic Hypergeometric Functions in Combinatorics, Integrable Systems and Physics'', at the Erwin Schr\"{o}dinger Institute, in Vienna, and APK thanks the participants and organisers, particularly V.P.~Spiridonov, for their comments. We also thank the anonymous referees for helpful comments, which led us to use a change of variables to write the~$A_n$ transformation in a way where the right hand side of~\eqref{transdef} is manifestly periodic in both the complex and integer variables. Particularly, the periodicity is necessary to allow the $(\textrm{mod }2r)$, and $(\textrm{mod }r)$, in the balancing condition~\eqref{balancing}.

MY would like to thank Harvard university for hospitality where part of this work was performed.
APK is an overseas researcher under Postdoctoral Fellowship of Japan Society for the Promotion of Science (JSPS).
MY is supported by WPI program (MEXT, Japan), by JSPS Program for Advancing Strategic International Networks to Accelerate the Circulation of Talented Researchers, by JSPS KAKENHI Grant No.~15K17634, and by JSPS-NRF research fund.

\pdfbookmark[1]{References}{ref}
\LastPageEnding

\end{document}